\documentclass[a4paper,UKenglish,10pt]{article}

\usepackage{microtype}
\usepackage{amscd}
\usepackage{pdfsync} 
\usepackage{cmll}
\usepackage{amssymb,amsmath}
\usepackage{proof}
\usepackage{url}

\newcommand{\midd}{\; \; \mbox{\Large{$\mid$}}\;\;}

\newenvironment{varitemize}
{
\begin{list}{\labelitemi}
{\setlength{\itemsep}{0pt}
 \setlength{\topsep}{0pt}
 \setlength{\parsep}{0pt}
 \setlength{\partopsep}{0pt}
 \setlength{\leftmargin}{15pt}
 \setlength{\rightmargin}{0pt}
 \setlength{\itemindent}{0pt}
 \setlength{\labelsep}{5pt}
 \setlength{\labelwidth}{10pt}
}}
{
 \end{list} 
}

\newcounter{numberone}

\newenvironment{varenumerate}
{
\begin{list}{\arabic{numberone}.}
{
  \usecounter{numberone}
  \setlength{\itemsep}{0pt}
  \setlength{\topsep}{0pt}
  \setlength{\parsep}{0pt}
  \setlength{\partopsep}{0pt}
  \setlength{\leftmargin}{15pt}
  \setlength{\rightmargin}{0pt}
  \setlength{\itemindent}{0pt}
  \setlength{\labelsep}{5pt}
  \setlength{\labelwidth}{15pt}
}}
{
\end{list} 
}

\newcounter{numbertwo}

\newcommand{\NN}{\mathbb{N}}

\newcommand{\RR}{\mathbb{R}}

\newcommand{\distrone}{\mathcal{D}}

\newcommand{\distrset}[1]{\mathbb{P}_{#1}}

\usepackage{units}
\usepackage{mathabx}
\usepackage{bussproofs}
\usepackage{stmaryrd}
\usepackage{a4wide}

\newtheorem{myth}{Theorem}

\newtheorem{definition}{Definition}
\newtheorem{lemma}{Lemma}

\newtheorem{theorem}{Theorem}
\newtheorem{corollary}{Corollary}
\newenvironment{proof}{\begin{trivlist}
	\item[\hskip \labelsep {\bfseries Proof.}]}{\hfill $\Box$ \end{trivlist}}

\newcommand{\bnf}{::=}

\usepackage{mathtools}
\DeclarePairedDelimiter\myabs{\lvert}{\rvert}

\newcommand{\types}{\mathrm A}
\newcommand{\typestring}{\textsf{Str}}

\newcommand{\typeone}{\textsf A}
\newcommand{\typetwo}{\textsf B}
\newcommand{\typethree}{\textsf C}
\newcommand{\typefour}{\textsf D}
\newcommand{\aone}{\textsf a}
\newcommand{\atwo}{\textsf b}
\newcommand{\athree}{\textsf c}
\newcommand{\astring}{\textsf s}
\newcommand{\modal}{\boxempty}
\newcommand{\nonmodal}{\blacksquare}

\newcommand{\terms}{\mathrm T}
\newcommand{\emptys}{\epsilon}

\newcommand{\tone}{t}
\newcommand{\ttwo}{s}
\newcommand{\tthree}{r}
\newcommand{\tfour}{q}

\newcommand{\conzero}{{\mathsf{0}}}
\newcommand{\conone}{{\mathsf{1}}}
\newcommand{\tail}{\mathsf{tail}}
\newcommand{\rand}{\mathsf{rand}}
\newcommand{\ccase}[5]{\textsf{case}_{#1}(#2,#3,#4,#5)}

\newcommand{\rrec}[5]{\textsf{rec}_{#1}(#2,#3,#4,#5)}

\newcommand{\varone}{x}
\newcommand{\vartwo}{y}
\newcommand{\varthree}{z}
\newcommand{\varfour}{w}
\newcommand{\vars}{\mathrm X}
\newcommand{\subs}[2]{\{\nicefrac{#1}{#2}\}}
\newcommand{\zero}{{\mathsf 0}}
\newcommand{\one}{{\mathsf 1}}

\newcommand{\values}{\mathrm V}
\newcommand{\vone}{v}

\newcommand{\free}[1]{\{#1\}}

\newcommand{\tcontone}{\Gamma}
\newcommand{\tconttwo}{\Delta}

\newcommand{\emtcont}{\emptyset}
\newcommand{\proves}{\vdash}

\newcommand{\conv}{\Downarrow}
\newcommand{\done}{\mathcal D}
\newcommand{\dtwo}{\mathcal E}
\newcommand{\dthree}{\mathcal F}
\newcommand{\pone}{\alpha}

\newcommand{\supp}{\textsf{S}}

\newcommand{\sem}[1]{\llbracket #1 \rrbracket}

\newcommand{\mcstates}{\mathrm S}
\newcommand{\mclabels}{\mathsf L}
\newcommand{\mcprob}{\mathcal P}
\newcommand{\relone}{\mathcal{R}}

\newcommand{\eval}{\textsf{eval}}
\newcommand{\pass}[1]{\textsf{pass}(#1)}

\newcommand{\pabone}{\mathcal R}
\newcommand{\eqcone}{E}
\newcommand{\labone}{l}

\newcommand{\open}[1]{#1_\circ}

\newcommand{\howe}[1]{#1^H}
\newcommand{\tc}[1]{(#1)^+}

\newcommand{\How}[1]{\mathsf{H_#1}}
\newcommand{\TC}[1]{\mathsf{TC_#1}}

\newcommand{\hopen}[1]{\howe{\open #1}}

\newcommand{\hotc}[1]{\tc{\hopen #1}}

\newcommand{\contone}{C}

\newcommand{\varsec}{\textsf{sec}}

\newcommand{\distone}{d}

\newcommand{\mcone}{\mathcal M}

\newcommand{\neglone}{\varepsilon}
\newcommand{\Nat}{\mathbb{N}}
\newcommand{\Real}{\mathbb{R}}
\newcommand{\secp}{\mathbf{n}}
\newcommand{\bit}{\mathrm{b}}
\newcommand{\strone}{{m}}
\newcommand{\strtwo}{{n}}
\newcommand{\strp}[1]{\underline{#1}}
\newcommand{\strset}{\mathrm{M}}

\newcommand{\view}[1]{\textsf{view}(#1)}
\newcommand{\trone}{\mathsf{T}}
\newcommand{\trtwo}{\mathsf{S}}
\newcommand{\emptyt}{{\mathbf\epsilon}}

\newcommand{\treq}{\simeq^{\mathsf T}}
\newcommand{\hole}{[\cdot]}

\newcommand{\tdone}{\mathcal T}
\newcommand{\tdtwo}{\mathcal S}
\newcommand{\tdthree}{\mathcal U}
\newcommand{\pdone}{\mathcal P}
\newcommand{\pdtwo}{\mathcal Q}
\newcommand{\pdthree}{\mathcal M}
\newcommand{\pdfour}{\mathcal N}

\newcommand{\ppone}{p}
\newcommand{\pptwo}{q}
\newcommand{\ppthree}{r}

\newcommand{\onestep}{\rightarrow}
\newcommand{\pairconv}{\Rrightarrow}
\newcommand{\smallstep}[1]{\Rightarrow^{#1}}
\newcommand{\sister}{\triangledown}

\newcommand{\tramet}{\delta^{\mathtt T}}
\newcommand{\conmet}{\delta^{\mathtt C}}

\newcommand{\relative}[1]{\triangledown_{#1}}

\newcommand{\secs}{{1^\secp}}

\newcommand{\ptreq}{\treq_n}
\newcommand{\ptdone}{\bar\tdone}
\newcommand{\ptdtwo}{\bar\tdtwo}

\newcommand{\ddone}{\mathtt D}

\newcommand{\rfun}[1]{\mathtt{f}_{#1}}
\newcommand{\rbg}{\mathbf{\mathtt{RBG}}}

\newcommand{\level}{\mathbf{\mathtt{LV}}}
\newcommand{\ite}[3]{\mathsf{if}\ #1\ \mathsf{then}\ #2\ \mathsf{else}\ #3}

\newcommand{\T}{\ensuremath{\mathsf{T}}}
\newcommand{\SLR}{\ensuremath{\mathsf{SLR}}}
\newcommand{\RSLR}{\ensuremath{\mathsf{RSLR}}}

\title{On Equivalences, Metrics, and Polynomial Time\\(Long Version)\footnote{This work is partially supported by the ANR
    project 12IS02001 PACE.}}
\author{Alberto Cappai\footnote{Universit\`a di Bologna \& INRIA Sophia Antipolis} \and 
        Ugo Dal Lago\footnote{Universit\`a di Bologna \& INRIA Sophia Antipolis}}
\date{}

\begin{document}

\maketitle

\begin{abstract}
Interactive behaviors are ubiquitous in modern cryptography, but are
also present in $\lambda$-calculi, in the form of higher-order
constructions. Traditionally, however, typed $\lambda$-calculi simply
do not fit well into cryptography, being both deterministic and too
powerful as for the complexity of functions they can express.  We
study interaction in a $\lambda$-calculus for probabilistic polynomial
time computable functions. In particular, we show how notions of
context equivalence and context metric can both be characterized by
way of traces when defined on linear contexts. We then give evidence
on how this can be turned into a proof methodology for computational
indistinguishability, a key notion in modern cryptography. We also
hint at what happens if a more general notion of a context is used.
\end{abstract}
\section{Introduction}
Modern cryptography~\cite{KatzLindell2007} is centered around the idea
that security of cryptographic constructions needs to be defined
precisely and, in particular, that crucial aspects are \emph{how} an
adversary interacts with the construction, and \emph{when} he wins
this game. The former is usually specified by way of an
\emph{experiment}, while the latter is often formulated stipulating that
the probability of a favorable result for the adversary needs to be
small, where being ``small'' usually means being \emph{negligible} in
a security parameter. This framework would however be vacuous if the
adversary had access to an unlimited amount of resources, or if it
were deterministic. As a consequence the adversary is usually assumed
to work within probabilistic polynomial time (PPT in the following),
this way giving rise to a robust definition. Summing up, there are
three key concepts here, namely
\emph{interaction}, \emph{probability} and \emph{complexity}.
Security as formulated above can often be spelled out semantically as
the so-called \emph{computational indistinguishability} between two
distributions, the first one being the one produced by the
construction and the second one modeling an idealized construction or
a genuinely random object.

Typed $\lambda$-calculi as traditionally conceived, do not fit well
into this picture. Higher-order types clearly allow a certain degree
of interaction, but probability and complexity are usually absent:
reduction is deterministic (or at least confluent), while the
expressive power of $\lambda$-calculi tends to be very high. This
picture has somehow changed in the last ten years: there have been
some successful attempts at giving probabilistic $\lambda$-calculi
whose representable functions coincide with the ones which can be
computed by PPT
algorithms~\cite{MitchellMS98,Zhang10,DalLagoParisenToldin}. These
calculi invariably took the form of restrictions on G\"odel's \T,
endowed with a form of binary probabilistic choice. All this has been
facilitated by implicit computational complexity, which offers the
right idioms to start from~\cite{Hofmann97}, themselves based on
linearity and ramification. The emphasis in all these works were
either the characterization of probabilistic complexity
classes~\cite{DalLagoParisenToldin}, or more often
security~\cite{Zhang10,NowakZ10,NowakZ14}: one could see
$\lambda$-calculi as a way to specify cryptographic constructions and
adversaries for them. The crucial idea here is that computational
indistinguishability can be formulated as a form of context
equivalence. The real challenge, however, is whether all this can be
characterized by handier notions, which would alleviate the inherently
difficult task of dealing with all contexts when proving two terms to
be equivalent.

The literature offers many proposals going precisely in this
direction: this includes logical relations, context lemmas, or
coinductive techniques. In applicative bisimulation~\cite{Abramsky90},
as an example, terms are modeled as interactive objects. This way, one
focuses on how the interpreted program interacts with its environment,
rather than on its internal evolution. None of them have so far been
applied to calculi capturing probabilistic polynomial time, and
relatively few among them handle probabilistic behavior.

In this paper, we study notions of equivalence and distance in one of
these $\lambda$-calculi,
called \RSLR~\cite{DalLagoParisenToldin}. More precisely:
\begin{varitemize}
\item
  After having briefly introduced \RSLR\ and studied its basic
  metatheoretical properties (Section~\ref{sect:characterizing}), we
  define \emph{linear context equivalence}. We then show how the role
  of contexts can be made to play by \emph{traces}. Finally, a
  coinductive notion of equivalence in the style of Abramsky's
  bisimulation is shown to be a congruence, thus included in context
  equivalence, but not to coincide with it. We also hint at how all
  this can be extended to metrics. This can be found in
  Section~\ref{sect:metrics}.
\item
  We then introduce a notion of \emph{parametrized context
    equivalence} for \RSLR\ terms, showing that it coincides with
    computational indistinguishability when the compared programs are
    of base type. We then turn our attention to the problem of
    characterizing the obtained notion of equivalence by way of linear
    tests, giving a positive answer to that by way of a notion of
    parametrized trace metric. A brief discussion about the role of
    linear contexts in cryptography is also given. All this is in
    Section~\ref{sect:ci}.
\end{varitemize}

\section{Characterizing Probabilistic Polynomial Time}\label{sect:characterizing}
In this section we introduce \RSLR~\cite{DalLagoParisenToldin}, a
$\lambda$-calculus for probabilistic polynomial time computation,
obtained by extending Hofmann's \SLR~\cite{Hofmann00} with an operator
for binary probabilistic choice. Compared to other presentations of
the same calculus, we consider a call-by-value reduction but elide
nonlinear function spaces and pairs. This has the advantage of making
the whole theory less baroque, without any fundamental loss in
expressiveness (see Section~\ref{sect:hoci} below).

First of all, \emph{types} are defined as follows:
$$
\typeone\bnf\ \typestring\midd\nonmodal\typeone\to\typeone\midd\modal\typeone\to\typeone.
$$
The expression $\typestring$ serves to type strings, and is the only
base type. $\nonmodal\typeone\to\typetwo$ is the type of functions
(from $\typeone$ to $\typetwo$) which can be evaluated in constant
time, while for $\modal\typeone\to\typetwo$ the running time can be
any polynomial. \emph{Aspects} are the elements of
$\{\modal,\nonmodal\}$ and are indeed fundamental to ensure polytime
soundness. We denote them with metavariables like $\aone$ or $\atwo$.
We define a partial order $<:$ between aspects simply as
$\{(\modal,\modal),(\modal,\nonmodal),(\nonmodal,\nonmodal)\}$, and a
subtyping by using the rules in Figure~\ref{fig:subtyping}.
\begin{figure}[!ht]
\centering
\fbox{
  \begin{minipage}{.95\textwidth}
  $$
  \AxiomC{}
  \UnaryInfC{$\typeone<:\typeone$}
  \DisplayProof
  \qquad
  \AxiomC{$\typeone<:\typetwo$}
  \AxiomC{$\typetwo<:\typethree$}
  \BinaryInfC{$\typeone<:\typethree$}
  \DisplayProof
  \qquad
  \AxiomC{$\typetwo<:\typeone$}
  \AxiomC{$\typethree<:\typefour$}
  \AxiomC{$\aone<:\atwo$}
  \TrinaryInfC{$\aone\typeone\to\typethree<:\atwo\typetwo\to\typefour$}
  \DisplayProof
   $$
  \end{minipage}
}\caption{Subtyping Rules}\label{fig:subtyping}
\end{figure}

The syntactical categories of \emph{terms} and \emph{values} are
constructed by the following grammar:
{\footnotesize
\begin{align*} 
\tone\ \bnf\ & \varone\midd\vone\midd\conzero(\tone)\midd\conone(\tone)\midd\tail(\tone)\midd\tone\tone\midd
		\ccase{\typeone}{\tone}{\tone}{\tone}{\tone}\midd
		\rrec{\typeone}{\tone}{\tone}{\tone}{\tone}\midd\rand;\\
\vone\ \bnf\ & \strp\strone\midd\lambda\varone:\aone\typeone.\tone;
\end{align*}}
where $\strone$ ranges over the set $\{0,1\}^*$ of finite, binary
strings, while $\varone$ ranges over a denumerable set of variables
$\vars$. We write $\terms,\values$ for the sets of terms
and values, respectively. The operators $\conzero$ and $\conone$ are
constructors for binary strings, while $\tail$ is a destructor. The
only nonstandard constant is $\rand$, which returns $\strp{0}$ or $\strp{1}$, each
with probability $\frac{1}{2}$, thus modeling uniform binary choice.
The terms $\ccase{\typeone}{\tone}{\tone_0}{\tone_1}{\tone_\epsilon}$
and $\rrec{\typeone}{\tone}{\tone_0}{\tone_1}{\tone_\epsilon}$ are
terms for case distinction and recursion, in which first argument
specifies the term (of base type) which guides the process.
Informally, then, we have the following rules:
\begin{align*}
	\ccase{\typeone}{\strp\emptys}{\tone_0}{\tone_1}{\tone_\epsilon}&\to\tone_\epsilon; &
	\rrec{\typeone}{\strp\emptys}{\tone_0}{\tone_1}{\tone_\epsilon}&\to\tone_\epsilon\\
	\ccase{\typeone}{\strp{\zero\strone}}{\tone_0}{\tone_1}{\tone_\epsilon}&\to\tone_0; &
	\rrec{\typeone}{\strp{\zero\strone}}{\tone_0}{\tone_1}{\tone_\epsilon}&\to(\tone_0\strp{\zero\strone})(\rrec{\typeone}{\strp\strone}{\tone_0}{\tone_1}{\tone_\epsilon})\\
	\ccase{\typeone}{\strp{\one\strone}}{\tone_0}{\tone_1}{\tone_\epsilon}&\to\tone_1;&
        \rrec{\typeone}{\strp{\one\strone}}{\tone_0}{\tone_1}{\tone_\epsilon}&\to(\tone_1\strp{\one\strone})(\rrec{\typeone}{\strp\strone}{\tone_0}{\tone_1}{\tone_\epsilon})
\end{align*}
\newcommand{\tailset}[1]{\mathbf{T}(#1)}
The expression $\strp\emptys$ stands for the empty string and we set
$\tail(\strp\emptys)\onestep\strp\emptys$. Given a string $\strp\strone$, $\tailset{\strp\strone}$ is
the set of strings whose tail is $m$,
e.g. $\tailset{\epsilon}=\{\epsilon,0,1\}$.

As usual, a \emph{typing context} $\tcontone$ is a finite set of
assignments of an aspect and a type to a variable, where as usual any
variable occurs \emph{at most once}. Any such assignment is indicated
with $\varone:\aone\typeone$. The expression $\tcontone,\tconttwo$ stands for the union of
the two typing contexts $\tcontone$ and $\tconttwo$, which are assumed to be
disjoint. The union $\tcontone,\tconttwo$ is indicated with
$\tcontone;\tconttwo$ whenever we want to insist on $\tcontone$ to only
involve the \emph{base} type $\typestring$. \emph{Typing judgments} are in the form
$\tcontone\proves\tone:\typeone$. Typing rules are in Figure~\ref{fig_typingrules}. 
\begin{figure}[!h]
\centering \fbox {
\begin{minipage}{.98\textwidth}
\footnotesize
{
 	$$
 	\AxiomC{$\varone:\aone\typeone\in\tcontone$}
 	\UnaryInfC{$\tcontone\proves\varone:\typeone$}
 	\DisplayProof
        \quad
        \AxiomC{}
        \UnaryInfC{$\tcontone\proves\strp\strone:\typestring$}
        \DisplayProof  
        \quad  
        \AxiomC{$\tcontone\proves\tone:\typestring$}
        \UnaryInfC{$\tcontone\proves\conzero(\tone):\typestring$}
        \DisplayProof
        \quad  
        \AxiomC{$\tcontone\proves\tone:\typestring$}
        \UnaryInfC{$\tcontone\proves\conone(\tone):\typestring$}
        \DisplayProof
 	\quad
        \AxiomC{$\tcontone\proves\tone:\typestring$}
        \UnaryInfC{$\tcontone\proves\tail(\tone):\typestring$}
        \DisplayProof
        \quad  
        \AxiomC{}
        \UnaryInfC{$\proves\rand:\typestring$}
        \DisplayProof
 	$$
        \vspace{6pt}
 	$$
 	\AxiomC{$
        \begin{array}{cc}
        \tcontone;\tconttwo_1\proves\tone:\typestring &
        \tcontone;\tconttwo_3\proves\tone_1:\typeone\\
	\tcontone;\tconttwo_2\proves\tone_0:\typeone &
 	\tcontone;\tconttwo_4\proves\tone_\epsilon:\typeone 
         \end{array}$}
 	\UnaryInfC{$\tcontone;\tconttwo_1,\tconttwo_2,\tconttwo_3,\tconttwo_4\proves\ccase{\typeone}{\tone}{\tone_0}{\tone_1}{\tone_\epsilon}:\typeone$}
 	\DisplayProof
        \qquad  
 	\AxiomC{$\tcontone,\varone:\aone\typeone\proves\tone:\typetwo$}
 	\UnaryInfC{$\tcontone\proves\lambda\varone:\aone\typeone.\tone:\aone\typeone\to\typetwo$}
 	\DisplayProof
        \qquad  
 	\AxiomC{$\tcontone\proves\tone:\typeone$}
 	\AxiomC{$\typeone<:\typetwo$}
 	\BinaryInfC{$\tcontone\proves\tone:\typetwo$}
 	\DisplayProof
 	$$
        \vspace{6pt}       
 	$$
 	\AxiomC{
        $
        \begin{array}{cc}
        \tcontone_1;\tconttwo_1\proves\tone:\typestring &  
        \tcontone_1,\tcontone_2,\tcontone_3;\tconttwo_2\proves\tone_\epsilon:\typeone\\
        \tcontone_1,\tcontone_2\proves\tone_0:\modal\typestring\to\nonmodal\typeone\to\typeone &
        \tcontone_1,\tconttwo_1<:\modal\\
        \tcontone_1,\tcontone_3\proves\tone_1:\modal\typestring\to\nonmodal\typeone\to\typeone &
        \mbox{$\typeone$ is $\modal$-free}
        \end{array}
        $}
 	\UnaryInfC{$\tcontone_1,\tcontone_2,\tcontone_3;\tconttwo_1,\tconttwo_2\proves
          \rrec{\typeone}{\tone}{\tone_0}{\tone_1}{\tone_\epsilon}:\typeone$}
 	\DisplayProof
 	\quad
 	\AxiomC{$\tcontone;\tconttwo_1\proves\tone:\aone\typeone\to\typetwo$}
 	\noLine
 	\UnaryInfC{$\tcontone;\tconttwo_2\proves\ttwo:\typeone$}
 	\AxiomC{$\tcontone,\tconttwo_2<:\aone$}
 	\BinaryInfC{$\tcontone;\tconttwo_1,\tconttwo_2\proves\tone\ttwo:\typetwo$}
 	\DisplayProof
 	$$
}
\end{minipage}}
\caption{\RSLR's Typing Rules}\label{fig_typingrules}
\end{figure}
The expression $\terms^\typeone_\tcontone$ (respectively,
$\values^\typeone_\tcontone$) stands for the set of terms
(respectively, values) of type $\typeone$ under the typing context
$\tcontone$. Please observe how the type system we have just
introduced enforces variables of higher-order type to occur free at
most once and outside the scope of a recursion. Moreover, the type of
terms which serve as step-functions in a recursion are assumed to be
$\modal$-free, and this is precisely what allow this calculus to
characterize polytime functions.

The operational semantics of \RSLR\ is of course probabilistic: any
closed term $\tone$ evaluates not to a single value but to
a \emph{value distribution}, i.e, a function $\done:\values\to\RR$
such that $\sum_{\vone\in\values}\done(\vone)=1$. Judgments 
expressing this fact are in the form $\tone\conv\done$, and
are derived through a formal system whose rules are in Figure~\ref{fig:bigstepsem}.
In the figure, and in the rest of this paper, we use some standard
notation on distributions. More specifically, the expression
$\{\vone_1^{\pone_1},\ldots,\vone_n^{\pone_n}\}$ stands for the
distribution assigning probability $\pone_i$ to $\vone_i$ (for every
$1\leq i\leq n$). The support of a distribution $\done$
is indicated with $\supp(\done)$. Given a set $X$, $\distrset{X}$
is the set of all distributions over $X$.
\begin{figure}[!ht]
\centering
\fbox{
	\begin{minipage}{.95\textwidth}
	\footnotesize{
	$$
	\AxiomC{}
	\UnaryInfC{$\vone\conv\{\vone^1\}$}
	\DisplayProof
	\quad
	\AxiomC{}
	\UnaryInfC{$\rand\conv\{\underline{0}^{\frac12},\underline{1}^{\frac12}\}$}
	\DisplayProof
	\qquad
	\AxiomC{$\tone\conv\done$}
	\AxiomC{$\ttwo\conv\dtwo$}
	\AxiomC{$\{\tthree\subs{\vone}{\varone}\conv\dthree_{\tthree,\vone}\}_{\lambda\varone.\tthree,\vone}$}
	\TrinaryInfC{$\tone\ttwo\conv\sum_{\lambda\varone.\tthree,\vone}\done(\lambda\varone.\tthree)\cdot\dtwo(\vone)\cdot\dthree_{\tthree,\vone}$}
	\DisplayProof
	$$
	\vspace{6pt}
	$$
	\AxiomC{$\tone\conv\{(\strp{\strone_i})^{\ppone_i}\}$}
	\UnaryInfC{$\conzero(\tone)\conv\{(\strp{\zero\strone_i})^{\ppone_i}\}$}
	\DisplayProof
	\qquad
	\AxiomC{$\tone\conv\{(\strp{\strone_i})^{\ppone_i}\}$}
	\UnaryInfC{$\conone(\tone)\conv\{(\strp{\one\strone_i})^{\ppone_i}\}$}
	\DisplayProof
	\qquad
	\AxiomC{$\tone\conv\done$}
	\UnaryInfC{$\tail(\tone)\conv\{(\strp{\strone_i})^{\done(\tailset{\strp{\strone_i}})}\}$}
	\DisplayProof
	$$
	\vspace{6pt}
	$$
	\AxiomC{$\tone\conv\done\qquad\tone_0\conv\done_0\qquad\tone_1\conv\done_1\qquad\tone_\epsilon\conv\done_\epsilon$}
	\UnaryInfC{$\ccase{\typeone}{\tone}{\tone_0}{\tone_1}{\tone_\epsilon}\conv\sum_{\strp\strone}\done(\strp{\zero\strone})\cdot\done_0+\sum_{\strp\strone}\done(\strp{\one\strone})\cdot\done_0+\done(\strp\emptys)\cdot\done_\epsilon$}
	\DisplayProof
	$$
	\vspace{6pt}
	$$
	\AxiomC{$\tone\conv\done$}
	\noLine
	\UnaryInfC{$\tone_\epsilon\conv\done_{\strp\emptys}$}
	\AxiomC{$\{(\tone_0\strp\strone)(\rrec{\typeone}{\strp\strtwo}{\tone_0}{\tone_1}{\tone_\epsilon})\conv\done_{\strp\strone}\}_{\strp\strone=\strp{\zero\strtwo}}$}
	\noLine
	\UnaryInfC{$\{(\tone_1\strp\strone)(\rrec{\typeone}{\strp\strtwo}{\tone_0}{\tone_1}{\tone_\epsilon})\conv\done_{\strp\strone}\}_{\strp\strone=\strp{\one\strtwo}}$}
	\BinaryInfC{$\rrec{\typeone}{\tone}{\tone_0}{\tone_1}{\tone_\epsilon}\conv\sum_{\strp\strone}\done(\strp\strone)\done_{\strp\strone}$}
	\DisplayProof
	$$
	}
	\end{minipage}
}
\caption{Big-step Semantics}\label{fig:bigstepsem} 
\end{figure}
Noticeably:
\begin{lemma}
For every term $\tone\in\terms^\typeone_\emtcont$ there is 
a unique distribution $\done$ such that $\tone\conv\done$,
which we denote as $\sem{\tone}$. Moreover, If $\vone\in\supp(\done)$,
then $\vone\in\values^\typeone_\emtcont$.
\end{lemma}
\begin{proof}
We proceed by induction on the structure of $\tone$.
\begin{itemize}
\item
	If we have a value $\vone$, then by the rules it converge to $\{\vone^1\}$.
\item
	Similarly if we have a term $\rand$ the only distribution it can converge is $\{\strp\zero^{\frac12},\strp\one^{\frac12}\}$.
\item
	Suppose now to have $\tone_1\tone_2$, and suppose $\tone_1\tone_2\conv\done, \tone_1,\tone_2\conv\done'$.\\
	By construction we have: 
	$$\done=\sum_{\lambda\varone.\tone,\vone}\done_1(\lambda\varone.\tone)\cdot\done_2(\vone)\cdot\done_{\tone,\vone}\qquad\done'=\sum_{\lambda\varone.\tone',\vone'}\done_1'(\lambda\varone.\tone')\cdot\done_2'(\vone')\cdot\done_{\tone',\vone'}'$$
	But, by induction hypothesis we have $\done_1=\done_1',\done_2=\done_2'$ and so also $\done_{\tone,\vone}=\done_{\tone',\vone'}'$ and this means $\done=\done'$
\item
	All the other cases are similar. 
\end{itemize}
The second point comes from the fact that, given a term $\tone$ such that $\proves\tone:\typeone$, if it reduces to $\tone_1,...,\tone_n$, we have that $\proves\tone_i:\typeone$. This is proved by induction on the type derivation. So by combinig the fact that the type is preserved by reduction and the uniqueness of $\done$ we have that for all $\vone\in\supp(\done), \proves\vone:\typeone$.
\end{proof}

\newcommand{\pfone}{F} 
A \emph{probabilistic function} on $\{0,1\}^*$ is a function $\pfone$
from $\{0,1\}^*$ to $\distrset{\{0,1\}^*}$. A term
$\tone\in\terms^{\aone\typestring\to\typestring}_{\emtcont}$ is said
to \emph{compute} $\pfone$ iff for every string $\strp\strone\in\{0,1\}^*$ it
holds that $\tone\strp\strone\conv\distrone$ where
$\distrone(\strp\strtwo)=\pfone(\strp\strone)(\strp\strtwo)$ for every
$\strp\strtwo\in\{0,1\}^*$.  What makes \RSLR\ very interesting, however, is that
it precisely captures those probabilistic functions which can be
computed in polynomial time (see, e.g., \cite{LagoZG14} for a
definition):
\begin{theorem}[Polytime Completeness]
The set of probabilistic functions which can be computed
by \RSLR\ terms coincides with the polytime computable ones.
\end{theorem}
This result is well-known~\cite{Zhang10,DalLagoParisenToldin}, and can be proved in various
ways, e.g. combinatorially or categorically.

We conclude this section by giving two \RSLR\ programs. Both of them
receive a string in input. The first one returns the same string.
The second one, instead, produces a random string and compare it to the
one received in input; if they are different it returns the same
string, otherwise it returns the opposite.
\begin{align*}
	\tone:=\lambda\varone:\modal\typestring.\varone
	\qquad
	\ttwo:=\lambda\varone:\modal\typestring.\ccase{\typestring}{\varone=(\rbg\ \varone)}{\varone}{\neg\varone}{\neg\varone}
\end{align*}
Where:
\begin{align*}
	&
		\rbg:=\lambda\vartwo:\modal\typestring.\rrec{\typestring}{\vartwo}{\rfun\rbg}{\rfun\rbg}{\strp\emptys}
		\qquad
		\rfun\rbg:=\lambda\varfour:\modal\typestring.\lambda\varthree:\nonmodal\typestring.\ccase{\typestring}{\rand}{\conzero(\varthree)}{\conone(\varthree)}{\strp\emptys}&
\end{align*}
Notice that, even if we haven't defined = and $\neg$, they are easily implementable in \RSLR.\\\\
We give now a simple example of how the big step semantics of a $\RSLR$ term is evaluated; we observe the term $\rbg$ applied to a string $\strp{\zero\one}$.
\begin{align*}
\sem{\rbg\ \strp{\zero\one}} = &
	\sem{\rbg}(\lambda\vartwo.\rrec{\typestring}{\vartwo}{\rfun\rbg}{\rfun\rbg}{\strp\emptys})\cdot\sem{\strp{\zero\one}}(\strp{\zero\one})\cdot\sem{\rrec{\typestring}{\strp{\zero\one}}{\rfun\rbg}{\rfun\rbg}{\strp\emptys}}\\
	= &
	1\cdot1\cdot\sem{\rrec{\typestring}{\strp{\zero\one}}{\rfun\rbg}{\rfun\rbg}{\strp\emptys}} = \\
	= &
	\sem{\strp{\zero\one}}(\strp{\zero\one})\cdot\sem{(\rfun\rbg\strp{\zero\one})(\rrec{\typestring}{\strp\one}{\rfun\rbg}{\rfun\rbg}{\strp\emptys})} = \sem{(\rfun\rbg\strp{\zero\one})(\rrec{\typestring}{\strp\one}{\rfun\rbg}{\rfun\rbg}{\strp\emptys})}
\end{align*}
We can easily say that $\sem{\rfun\rbg\strp{\zero\one}}=\sem{\rfun\rbg\subs{\strp{\zero\one}}{\varfour}}=\{(\lambda\varthree.\ccase{\typestring}{\rand}{\conzero(\varthree)}{\conone(\varthree)}{\strp\emptys})^1\}$.\\
Furthermore we have:
\begin{align*}
\sem{\rrec{\typestring}{\strp\one}{\rfun\rbg}{\rfun\rbg}{\strp\emptys}} = &
	\sem{\strp\one}(\strp\one)\cdot\sem{(\rfun\rbg\strp\one)(\rrec{\typestring}{\strp\emptys}{\rfun\rbg}{\rfun\rbg}{\strp\emptys})}
\end{align*}
So, by the fact that $\sem{\rrec{\typestring}{\strp\emptys}{\rfun\rbg}{\rfun\rbg}{\strp\emptys}}=\{\emptys^1\}$ we have:
\begin{align*}
\sem{\rrec{\typestring}{\strp\one}{\rfun\rbg}{\rfun\rbg}{\strp\emptys}} = &
	\sem{\ccase{\typestring}{\rand}{\conzero(\strp\emptys)}{\conone(\strp\emptys)}{\strp\emptys}} = \sem{\rand}(\strp\zero)\cdot\sem{\conzero(\strp\emptys)}+\sem{\rand}(\strp\one)\cdot\sem{\conone(\strp\emptys)} = \{\strp\zero^{\frac12},\strp\one^{\frac12}\}
\end{align*}
So, by substituting we have:
\begin{align*}
\sem{\rbg\ \strp{\zero\one}} = &
	\sem{(\rfun\rbg\strp{\zero\one})(\rrec{\typestring}{\strp\one}{\rfun\rbg}{\rfun\rbg}{\strp\emptys})} =\\
	= &
	\frac12\cdot\sem{\ccase{\typestring}{\rand}{\conzero(\strp\zero)}{\conone(\strp\zero)}{\strp\emptys}}+\frac12\cdot\sem{\ccase{\typestring}{\rand}{\conzero(\strp\one)}{\conone(\strp\one)}{\strp\emptys}} =\\
	= &
	\frac12\cdot\{\strp{\zero\zero}^{\frac12},\strp{\one\zero}^{\frac12}\} + \frac12\cdot\{\strp{\zero\one}^{\frac12},\strp{\one\one}^{\frac12}\} =\\
	= &
	\{\strp{\zero\zero}^{\frac14},\strp{\zero\one}^{\frac14},\strp{\one\zero}^{\frac14},\strp{\one\one}^{\frac14}\}
\end{align*}
\section{Equivalences}\label{sect:equivalences}
Intuitively, we can say that two programs are equivalent if no one can
distinguish them by observing their external, visible, behavior. A
formalization of this intuition usually takes the form of context
equivalence. A
\emph{context} is a term in which the \emph{hole} $\hole$ 
occurs at most once. Formally, contexts are defined by the following grammar:
{\footnotesize
\begin{align*}
\contone\bnf&\ \tone\midd\hole\midd\lambda\varone.\contone\midd\contone\tone\midd\tone\contone\midd 
		 \midd\conzero(\contone)\midd\conone(\contone)\midd\tail(\contone)\\
            & \midd\ccase{\typeone}{\contone}{\tone}{\tone}{\tone}\midd\ccase{\typeone}{\tone}{\contone}{\contone}{\contone}\midd
             \rrec{\typeone}{\contone}{\tone}{\tone}{\tone}.
\end{align*}}
If the grammar above is extended as follows
$
\contone\bnf \rrec{\typeone}{\tone}{\contone}{\tone}{\tone}\midd
              \rrec{\typeone}{\tone}{\tone}{\contone}{\tone}\midd \rrec{\typeone}{\tone}{\tone}{\tone}{\contone},
$ 
what we get is a \emph{nonlinear} context.
What the above definition already tells us is that our emphasis in
this paper will be on \emph{linear} contexts, which are contexts whose
holes lie outside the scope of any recursion operator. Given a term
$\tone$ we define $\contone[\tone]$ as the term obtained by
substituting the occurrence of $\hole$ in $\contone$ (if any) with
$\tone$. We only consider non-binding contexts here, i.e. contexts are
meant to be filled with closed terms. In other words, the type system from
Section~\ref{sect:characterizing} can be turned into one for contexts
whose judgments take the form
$\tcontone\proves\contone[\proves\typeone]:\typetwo$, which means that
for every closed term $\tone$ of type $\typeone$, it holds that
$\tcontone\proves\contone[\tone]:\typetwo$. See Figure~\ref{fig:contexttypyngrules} for
  details.
\begin{figure}[!ht]
	\centering
	\fbox{
		\begin{minipage}{.95\textwidth}
		\footnotesize{
		$$
		\AxiomC{$\tcontone\proves\tone:\typeone$}
		\UnaryInfC{$\tcontone\proves\tone[\emptyset]:\typeone$}
		\DisplayProof
		\qquad
		\AxiomC{}
		\UnaryInfC{$\proves[\proves\typeone]:\typeone$}
		\DisplayProof
		\qquad
		\AxiomC{$\tcontone\proves\contone[\proves\typeone]:\typestring$}
		\UnaryInfC{$\tcontone\proves\conzero(\contone[\proves\typeone]),\conone(\contone[\proves\typeone]),\tail(\contone[\proves\typeone]):\typestring$}
		\DisplayProof
		$$
		\vspace{6pt}
		$$
		\AxiomC{$\varone:\atwo\typetwo,\tcontone\proves\contone[\proves\typeone]:\typethree$}
		\UnaryInfC{$\tcontone\proves\lambda\varone.\contone[\proves\typeone]:\atwo\typetwo\to\typethree$}
		\DisplayProof
		\qquad
		\AxiomC{$\begin{array}{cc}
		\tcontone;\tconttwo_1\proves\contone[\proves\typeone]:\atwo\typetwo\to\typethree &\\
		\tcontone;\tconttwo_2\proves\tone:\typetwo & \tcontone,\tconttwo_2<:\atwo
		\end{array}$}
		\UnaryInfC{$\tcontone;\tconttwo_1,\tconttwo_2\proves\contone\tone[\proves\typeone]:\typethree$}
		\DisplayProof
		$$
		\vspace{6pt}
		$$
		\AxiomC{$\begin{array}{cc}
		\tcontone;\tconttwo_1\proves\tone:\atwo\typetwo\to\typethree &\\
		\tcontone;\tconttwo_2\proves\contone[\proves\typeone]:\typetwo & \tcontone,\tconttwo_2<:\atwo
		\end{array}$}
		\UnaryInfC{$\tcontone;\tconttwo_1,\tconttwo_2\proves\tone\contone[\proves\typeone]:\typethree$}
		\DisplayProof
		\qquad
		\AxiomC{$\begin{array}{cc}
		\tcontone;\tconttwo_1\proves\contone[\proves\typeone]:\typestring & \tcontone;\tcontone_3\proves\tone_1:\typetwo \\
		\tcontone;\tconttwo_2\proves\tone_0:\typetwo & \tcontone;\tconttwo_4\proves\tone_\epsilon:\typetwo
		\end{array}$}
		\UnaryInfC{$\tcontone;\tconttwo_1,\tconttwo_2,\tconttwo_3,\tconttwo_4\proves\ccase{\typetwo}{\contone}{\tone_0}{\tone_1}{\tone_\epsilon}[\proves\typeone]:\typetwo$}
		\DisplayProof
		$$
		\vspace{6pt}
		$$
		\AxiomC{$\begin{array}{cc}
		\tcontone;\tconttwo_1\proves\tone:\typestring & \tcontone;\tconttwo_3\proves\contone_1[\proves\typeone]:\typetwo \\
		\tcontone;\tconttwo_2\proves\contone_0[\proves\typeone]:\typetwo & \tcontone;\tconttwo_4\proves\contone_\epsilon[\proves\typeone]:\typetwo
		\end{array}$}
		\UnaryInfC{$\tcontone;\tconttwo_1,\tconttwo_2,\tconttwo_3,\tconttwo_4\proves\ccase{\typetwo}{\tone}{\contone_0}{\contone_1}{\contone_\epsilon}[\proves\typeone]:\typetwo$}
		\DisplayProof
		$$
		\vspace{6pt}
		$$
		\AxiomC{$\begin{array}{cc}
		\tcontone_1;\tconttwo_1\proves\contone[\proves\typeone]:\typestring & \tcontone_1,\tcontone_2;\tcontone_3;\tconttwo_2\proves\tone_\epsilon:\typetwo\\
		\tcontone_1,\tcontone_2\proves\tone_0:\modal\typestring\to\nonmodal\typetwo\to\typetwo & \tcontone_1,\tconttwo_1<:\modal\\
		\tcontone_1,\tcontone_3\proves\tone_1:\modal\typestring\to\nonmodal\typetwo\to\typetwo & \typetwo\ \modal\text{-free}
		\end{array}$}
		\UnaryInfC{$\tcontone_1,\tcontone_2,\tcontone_3;\tconttwo_1,\tconttwo_2\proves\rrec{\typetwo}{\contone}{\tone_0}{\tone_1}{\tone_\epsilon}[\proves\typeone]:\typetwo$}
		\DisplayProof
		$$
		}
		\end{minipage}
	}
	\caption{Context Typing Rules}\label{fig:contexttypyngrules}
\end{figure}
\newcommand{\lceq}{\equiv}
\newcommand{\ceq}{\equiv_{\neg\ell}}
Now that the notion of a context has been properly defined, one can
finally give the central notion of equivalence in this paper.
\begin{definition}[Context Equivalence]\label{def:ceq}
  Given two terms $\tone,\ttwo$ such that
  $\proves\tone,\ttwo:\typeone$, we say that $\tone$ and $\ttwo$ are
  \emph{context equivalent} iff for every context $\contone$ such that
  $\proves\contone[\proves\typeone]:\typestring$ we have that
  $\sem{\contone[\tone]}(\strp\emptys)=\sem{\contone[\ttwo]}(\strp\emptys)$.
\end{definition}
The way we defined it means that context equivalence is a family of
relations $\{\lceq_\typeone\}_{\typeone\in\types}$ indexed by types,
which we denote as $\lceq$. If in Definition~\ref{def:ceq}
\emph{nonlinear} contexts replace contexts, we get a finer relation,
called \emph{nonlinear context equivalence}, which we denote as
$\ceq$. Both context equivalence and nonlinear context equivalence are
easily proved to be congruences, i.e. compatible equivalence relations.

\subsection{Trace Equivalence}\label{sect:traceequivalence}
In this section we introduce a notion of \textit{trace equivalence}
for \RSLR, and we show that it characterizes context equivalence.

We define a \textit{trace} as a sequence of actions
$\labone_1\cdot\labone_2\cdot\ldots\cdot\labone_n$ such that
$\labone_i\in\{\pass\vone,\view{\strp\strone}\
|\ \vone\in\values,\strp\strone\in\values^\typestring\}$. Traces are
indicated with metavariables like $\trone,\trtwo$.
The \emph{compatibility} of a trace $\trone$ with a type $\typeone$ is
defined inductively on the structure of $\typeone$. If
$\typeone=\typestring$ then the only trace compatible with $\typeone$
is $\trone=\view{\strp\strone}$, with $\strp\strone\in\values^\typestring$,
otherwise, if $\typeone=\atwo\typetwo\to\typethree$ then traces
compatible with $\typeone$ are in the form
$\trone=\pass{\vone}\cdot\trtwo$ with $\vone\in\values^\typetwo$ and
$\trtwo$ is itself compatible with $\typethree$. With a slight abuse
of notation, we often assume traces to be compatible to the underlying
type.

\newcommand{\smallstepreal}[1]{\mapsto^{#1}}
Due to the probabilistic nature of our calculus, it is convenient to
work with \emph{term distributions}, i.e., distributions whose support
is the set of closed terms of a certain type $\typeone$, instead of
plain terms. We denote term distributions with metavariables like
$\tdone,\tdtwo,\ldots$. The effect traces have to distributions can
be formalized by giving some binary relations:
\begin{varitemize}
\item
  First of all, we need a binary relation on term distributions,
  called $\pairconv$. Intuitively, $\tdone\pairconv\tdtwo$ iff
  $\tdone$ evolves to $\tdtwo$ by performing internal moves,
  only. Furthermore, we use $\onestep$ to indicate a single internal
  move.
\item 
  We also need a binary relation $\smallstep{\cdot}$ between
  term distributions, which is however labeled by a trace, and
  which models internal \emph{and external} reduction.\
\item
  Finally, we need a labeled relation $\smallstepreal{\cdot}$
  between distributions and \emph{real numbers}, which
  captures the probability that distributions accept
  traces.
\end{varitemize}
The three relations are defined inductively by the rules in
Figure~\ref{fig:tdsmallstep}.
\begin{figure}[!ht]
\centering \fbox{
\begin{minipage}{.97\textwidth}
    $$
    \AxiomC{}
    \UnaryInfC{$\tdone\smallstep\emptyt\tdone$}
    \DisplayProof
    \qquad
    \AxiomC{$\tdone\smallstep\trtwo\{(\lambda\varone.\tone_i)^{\ppone_i}\}$}
    \UnaryInfC{$\tdone\smallstep{\trtwo\cdot\pass\vone}\{(\tone_i\subs{\vone}{\varone})^{\ppone_i}\}$}
    \DisplayProof
    \qquad
    \AxiomC{$\tdone\smallstep\trtwo\tdtwo$}
    \AxiomC{$\tdtwo\pairconv\tdthree$}
    \BinaryInfC{$\tdone\smallstep\trtwo\tdthree$}
    \DisplayProof
    $$
    $$
    \AxiomC{$\tdone\smallstep\trtwo\{(\strp{\strone_i})^{\ppone_i}\}$}
    \UnaryInfC{$\tdone\smallstepreal{\trtwo\cdot\view{\strp\strone}}\sum_{\strp{\strone_i}=\strp\strone}\ppone_i$}
    \DisplayProof
    \qquad  
    \AxiomC{$\tone\onestep\{(\tone_i)^{\ppone_i}\}$}
    \UnaryInfC{$\tdone+\{(\tone)^\ppone\}\pairconv\tdone+\{(\tone_i)^{\ppone\cdot\ppone_i}\}$}
    \DisplayProof
    $$
\end{minipage}}\caption{Term Distribution Small-Step Rules}\label{fig:tdsmallstep}
\end{figure}
The following gives basic, easy, results about the relations we
have introduced:
\begin{lemma}
Let $\tdone$ be a term distribution for the type $\typeone$. Then, there is
a unique value distribution $\done$ such that $\tdone\pairconv^*\done$. As a
consequence, for every trace $\trone$ compatible for $\typeone$
there is a unique real number $\ppone$ such that
$\tdone\smallstepreal{\trone}\ppone$. This real number is
denoted as $\Pr(\tdone,\trone)$.
\end{lemma}
\begin{proof}
Suppose that $\tdone$ is normal, i.e. all elements in the support are values, then we have $\tdone=\done$ and then the thesis.\\
If $\tdone$ is not normal then there exists a set of indexes $J$ such that $\{(\tone_j)^{\ppone_j}\}_{j\in J}\subseteq\tdone$ aren't values. We know by a previous lemma that for all $j\in J$ there exists a unique $\done_j$, value distribution, such that $\tone_j\conv\done_j$ in a finite number of steps.\\
So, if we set $\done=\tdone\setminus\{(\tone_j)^{\ppone_j}\}_{j\in J}+\sum_{j\in J}\ppone_j\cdot\done_j$ we have $\tdone\pairconv^*\done$ with $\done$ normal.\\\\
At this point we can say that for all $\tdone$ there exists $\tdone'$ normal such that $\tdone\smallstep\emptyt\tdone'$. So, given $\trone=\trtwo\cdot\pass\vone$ we have by induction hypothesis that $\tdone\smallstep\trtwo\{(\lambda\varone.\tone_i)^{\ppone_i}\}$. Then, by performing the action $\pass\vone$ we have $\tdone\smallstep{\trtwo\cdot\pass\vone}\{(\tone_i\subs\vone\varone)^{\ppone_i}\}$, but, by applying the previous point there exists $\tdone'$ normal such that $\{(\tone_i\subs\vone\varone)^{\ppone_i}\}\pairconv^*\tdone'$ and then by the small step rules we have $\tdone\smallstep{\trtwo\cdot\pass\vone}\tdone'$ normal distribution.\\
Suppose now that $\trone=\trtwo\cdot\view{\strp\strone}$ then we have by induction hypothesis $\tdone\smallstep\trtwo\tdone'=\{(\strp{\strone_i})^{\ppone_i}\}$ with $\tdone'$ unique; so, if we perform the action $\view{\strp\strone}$ we have $\tdone\smallstepreal{\trtwo\cdot\view{\strp\strone}}\ppone=\sum_{i;\strp{\strone_i}=\strp\strone}\ppone_i$, that is unique by construction.
\end{proof}
%

We are now ready to define what we mean by trace equivalence
\begin{definition}
  Given two term distributions $\tdone,\tdtwo$ we say that they are
  \emph{trace equivalent} (and we write $\tdone\treq\tdtwo$) if, for
  all traces $\trone$ it holds that
  $\Pr(\tdone,\trone)=\Pr(\tdtwo,\trone)$. In particular, then, two
  terms $\tone,\ttwo$ are trace equivalent when
  $\{\tone^1\}\treq\{\ttwo^1\}$ and we write $\tone\treq\ttwo$ in
  that case.
\end{definition}
The following states some basic properties about the reduction relations
we have just introduced. This will be useful in the following:
\begin{lemma} \textbf{\textsf{(Trace Equivalence Properties): }}
Suppose given two term distributions $\tdone,\tdtwo$ such that
$\tdone\treq\tdtwo$. Then:
\begin{varitemize}
\item
  If $\tdone\pairconv\tdone'$ then
  $\tdone'\treq\tdtwo$.
\item
  If $\tdone\smallstep{\pass\vone}\tdone'$ and
    $\tdtwo\smallstep{\pass\vone}\tdtwo'$ then $\tdone'\treq\tdtwo'$.
\item
  If $\tdone\smallstepreal{\view{\strp\strone}}\ppone$ then
   $\tdtwo\smallstepreal{\view{\strp\strone}}\ppone$.
\end{varitemize}
\end{lemma}
\begin{proof}
The proof is a simple application of the definition of trace equivalence.
\end{proof}

It is easy to prove that trace equivalence is an equivalence
relation. The next step, then, is to prove that trace equivalence is
compatible, thus paving the way to a proof of soundness w.r.t.
context equivalence. Unfortunately, the direct proof of 
compatibility (i.e., an induction on the structure of contexts) simply
does not work: the way the operational semantics is specified makes it
impossible to track how a term behaves \emph{in a
context}. Following~\cite{DengZhang}, we proceed by considering
a refined semantics, defined not on terms but on pairs whose first
component is a context and whose second component is a term 
distribution.
Formally, a \emph{context pair} has the form $(\contone,\tdone)$,
where $\contone$ is a context and $\tdone$ is a term distribution.
A \emph{(context) pair distribution} is a distribution over context pairs.  Such
a pair distribution $\pdone=\{(\contone_i,\tdone_i)^{\ppone_i}\}$ is
said to be \emph{normal} if for all $i$ ad for all $\tone$ in the
support of $\tdone_i$ we have that $\contone_i[\tone]$ is a value. 
We show how a pair $(\contone,\tdone)$ evolves following a trace
$\trtwo$ by giving a one-step reduction relation (denoted with
$\onestep$) and the small-step semantic described in the rules in
Figure \ref{fig:onestep} and Figure \ref{fig:smallstep}.
\begin{figure}[!ht]
\centering \fbox{
\begin{minipage}{.97\textwidth}
\footnotesize
{
    $$
    \AxiomC{$\tdone\onestep^{\pass\vone}\tdone'$}
    \UnaryInfC{$(\hole,\tdone)\onestep^{\pass\vone}\{(\hole,\tdone')^1\}$}
    \DisplayProof
    \qquad
    \AxiomC{}
    \UnaryInfC{$(\lambda\varone.\contone,\tdone)\to^{\pass\vone}\{(\contone\subs\vone\varone,\tdone)^1\}$}
    \DisplayProof
    $$

    $$
    \AxiomC{}
    \UnaryInfC{$(\strp\strone,\tdone)\smallstepreal{\view{\strp\strone}}1$}
    \DisplayProof
    \qquad
    \AxiomC{}
    \UnaryInfC{$(\strp{\strone'},\tdone)\smallstepreal{\view{\strp\strone}}0$}
    \DisplayProof
    \qquad
    \AxiomC{$\tdone\smallstepreal{\view{\strp\strone}}\ppone$}
    \UnaryInfC{$(\hole,\tdone)\smallstepreal{\view{\strp\strone}}\ppone $}
    \DisplayProof
    $$
    
    $$
    \AxiomC{$(\contone,\tdone)\smallstepreal{\view{\strp\strone}}\ppone$}
    \UnaryInfC{$(\conzero(\contone),\tdone)\smallstepreal{\view{\strp{\zero\strone}}}\ppone$}
    \DisplayProof
    \qquad
    \AxiomC{$(\contone,\tdone)\smallstepreal{\view{\strp\strone}}\ppone$}
    \UnaryInfC{$(\conone(\contone),\tdone)\smallstepreal{\view{\strp{\one\strone}}}\ppone$}
    \DisplayProof
    \qquad
    \AxiomC{$\begin{array}{cc}
    (\contone,\tdone)\smallstepreal{\view{\strp{\bit\strone}}}\ppone_\bit &
    \bit\in\{0,1\}
    \end{array}$}
    \UnaryInfC{$(\tail(\contone),\tdone)\smallstepreal{\view{\strp\strone}}\ppone_0+\ppone_1$}
    \DisplayProof
    $$

    $$
    \AxiomC{$\tone\to\{(\tone_i)^{\ppone_i}\}$}
    \UnaryInfC{$(\tone,\tdone)\to\{(\tone_i,\tdone_i)^{\ppone_i}\}$}
    \DisplayProof
    \qquad
    \AxiomC{$\tdone\onestep\tdone'$}
    \UnaryInfC{$(\hole,\tdone)\onestep\{(\hole,\tdone')^1\}$}
    \DisplayProof
    $$

    $$
    \AxiomC{$(\contone,\tdone)\onestep^{\pass\vone}\{(\contone',\tdone')^1\}$}
    \UnaryInfC{$(\contone\vone,\tdone)\to\{(\contone',\tdone')^1\}$}
    \DisplayProof
    $$

    $$
    \AxiomC{$(\contone,\tdone)\to\{(\contone_i,\tdone_i)^{\ppone_i}\}$}
    \UnaryInfC{$(\contone\tone,\tdone)\to\{(\contone_i\tone,\tdone_i)^{\ppone_i}\}$}
    \DisplayProof
    \qquad
    \AxiomC{$\tone\to\{(\tone_i)^{\ppone_i}\}$}
    \AxiomC{$(\contone,\tdone)$ value}
    \BinaryInfC{$(\contone\tone,\tdone)\to\{(\contone\tone_i,\tdone_i)^{\ppone_i}\}$}
    \DisplayProof
    $$

    $$
    \AxiomC{$(\contone,\tdone)\to\{(\contone_i,\tdone_i)^{\ppone_i}\}$}
    \UnaryInfC{$(\vone\contone,\tdone)\to\{(\vone\contone_i,\tdone_i)^{\ppone_i}\}$}
    \DisplayProof
    \qquad
    \AxiomC{$\tone\to\{(\tone_i)^{\ppone_i}\}$}
    \UnaryInfC{$(\tone\contone,\tdone)\to\{(\tone_i\contone,\tdone_i)^{\ppone_i}\}$}
    \DisplayProof
    $$

    $$
    \AxiomC{}
    \UnaryInfC{$((\lambda\varone.\contone)\vone,\tdone)\to\{(\contone\subs{\vone}{\varone},\tdone)^1\}$}
    \DisplayProof
    \qquad
    \AxiomC{$(\contone,\tdone)\in\values$}
    \UnaryInfC{$((\lambda\varone.\tone)\contone,\tdone)\to\{(\tone\subs{\contone}{\varone},\tdone)^1\}$}
    \DisplayProof
    $$

    $$
    \AxiomC{$\tone\to\{(\tone_i)^{\ppone_i}\}$}
    \UnaryInfC{$(\ccase{\typeone}{\tone}{\contone_0}{\contone_1}{\contone_\epsilon},\tdone)
    \to\{(\ccase{\typeone}{\tone_i}{\contone_0}{\contone_1}{\contone_\epsilon},\tdone_i)^{\ppone_i}\}$}
    \DisplayProof
    $$

    $$
    \AxiomC{($\contone,\tdone)\to\{(\contone_i,\tdone_i)^{\ppone_i}\}$}
    \UnaryInfC{$(\ccase{\typeone}{\contone}{\tone_0}{\tone_1}{\tone_\epsilon},\tdone)\to
    \{(\ccase{\typeone}{\contone_i}{\tone_0}{\tone_1}{\tone_\epsilon},\tdone_i)^{\ppone_i}\}$}
    \DisplayProof
    $$

    $$
    \AxiomC{}
    \UnaryInfC{$(\ccase{\typeone}{\strp{\zero\strone}}{\contone_0}{\contone_1}{\contone_\epsilon},\tdone)\onestep
    \{(\contone_0,\tdone)^1\}$}
    \DisplayProof
    $$

    $$
    \AxiomC{}
    \UnaryInfC{$(\ccase{\typeone}{\strp{\one\strone}}{\contone_0}{\contone_1}{\contone_\epsilon},\tdone)\onestep
    \{(\contone_1,\tdone)^1\}$}
    \DisplayProof
    $$

    $$
    \AxiomC{}
    \UnaryInfC{$(\ccase{\typeone}{\strp\emptys}{\contone_0}{\contone_1}{\contone_\epsilon},\tdone)\onestep
    \{(\contone_\epsilon,\tdone)^1\}$}
    \DisplayProof
    $$
    
    $$
    \AxiomC{$\strp\strone\in\values^\typestring$}
    \AxiomC{$\{(\contone,\tdone)\smallstepreal{\view{\strp\strone}}\ppone_{\strp\strone}\}$}
    \BinaryInfC{$(\ccase{\typeone}{\contone}{\tone_0}{\tone_1}{\tone_\epsilon},\tdone)\onestep
    \{
     (\tone_0,\cdot)^{\sum_{\strp{\zero\strtwo}}\ppone_{\strp{\zero\strtwo}}},
    (\tone_1,\cdot)^{\sum_{\strp{\one\strtwo}}\ppone_{\strp{\one\strtwo}}},
    (\tone_\epsilon,\cdot)^{\ppone_{\strp\emptys}},
    \}$}
    \DisplayProof
    $$
    
    $$
    \AxiomC{$(\contone,\tdone)\onestep\{(\contone_i,\tdone_i)^{\ppone_i}\}$}
    \UnaryInfC{$(\rrec{\typeone}{\contone}{\tone_0}{\tone_1}{\tone_\epsilon},\tdone)\onestep
    	\{(\rrec{\typeone}{\contone_i}{\tone_0}{\tone_1}{\tone_\epsilon})^{\ppone_i},\tdone_i)^{\ppone_i}\}$}
    \DisplayProof
    $$
    
    $$
    \AxiomC{$(\contone,\tdone)\smallstepreal{\view{\strp\strone}}\ppone_{\strp\strone}$}
    \UnaryInfC{$(\rrec{\typeone}{\contone}{\tone_0}{\tone_1}{\tone_\epsilon})\onestep
    	\substack{
    	\{((\tone_0\strp\strone)\rrec{\typeone}{\strp\strtwo}{\tone_0}{\tone_1}{\tone_\epsilon},\tdone)^{\ppone_{\strp\strone}}\}_{\strp\strone=\strp{\zero\strtwo}} + \\
    	\{((\tone_1\strp\strone)\rrec{\typeone}{\strp\strtwo}{\tone_0}{\tone_1}{\tone_\epsilon},\tdone)^{\ppone_{\strp\strone}}\}_{\strp\strone=\strp{\one\strtwo}} + \{(\tone_\epsilon,\tdone)^{\ppone_{\strp\emptys}}\}}$}
    	\DisplayProof
    $$
}
\end{minipage}}\caption{One-step Rules}\label{fig:onestep}
\end{figure}
\begin{figure}[!ht]
\centering \fbox{
\begin{minipage}{.97\textwidth}
\footnotesize
{
	$$
	\AxiomC{}
	\UnaryInfC{$\pdone\smallstep\emptyt\pdone$}
	\DisplayProof
	\qquad
	\AxiomC{$\begin{array}{cc}
	\pdone\smallstep\trtwo\pdone' & \pdone'\pairconv\pdone''
	\end{array}$}
	\UnaryInfC{$\pdone\smallstep\trtwo\pdone''$}
	\DisplayProof
	$$
	\vspace{6pt}
	$$
	\AxiomC{$\begin{array}{cc}
	\pdone\smallstep\trtwo\{(\contone_i,\tdone_i)^{\ppone_i}\} & (\contone_i,\tdone_i)\onestep^{\pass\vone}\{(\contone_i',\tdone_i')^1\}
	\end{array}$}
	\UnaryInfC{$\pdone\smallstep{\trtwo\cdot\pass\vone}\{(\contone_i',\tdone_i')^{\ppone_i}\}$}
	\DisplayProof
	$$
	\vspace{6pt}
	$$
	\AxiomC{$\begin{array}{cc}\pdone\smallstep\trtwo\{(\contone_i,\tdone_i)^{\ppone_i}\} & 
	(\contone_i,\tdone_i)\smallstepreal{\view{\strp\strone}}\ppone_i'
	\end{array}$}
	\UnaryInfC{$\pdone\smallstepreal{\trtwo\cdot\view{\strp\strone}}\sum_i\ppone_i\cdot\ppone_i'$}
	\DisplayProof
	\qquad
	\AxiomC{$(\contone,\tdone)\onestep\{(\contone_i,\tdone_i)^{\ppone_i}\}$}
	\UnaryInfC{$\pdone+\{(\contone,\tdone)^\ppone\}\pairconv\pdone+\{(\contone_i,\tdone_i)^{\ppone\cdot\ppone_i}\}$}
	\DisplayProof
	$$
}
\end{minipage}}\caption{Small-Step Rules}\label{fig:smallstep}
\end{figure}

The following tells us that working with context pairs is the same as
working with terms as far as traces are concerned:
\begin{lemma}\label{lemma:pairconv}
  Suppose given a context $\contone$, a term distribution $\tdone$,
  and a trace $\trtwo$. Then if
  $(\contone,\tdone)\smallstep\trtwo\{(\contone_i,\tdone_i)^{\ppone_i}\}$
  then
  $\contone[\tdone]\smallstep\trtwo\{(\contone_i[\tdone_i])^{\ppone_i}\}$.
  Moreover, if $(\contone,\tdone)\smallstepreal\trtwo\ppone$, then
  $\Pr(\contone[\tdone],\trtwo)=\ppone$.
\end{lemma}
\begin{proof} 
\begin{itemize}
\item
    The first case comes from the definition of 1-step and small-step
    semantics.
\item
    If $\trtwo=\trtwo'\cdot\view{\strp\strone}$ with $\trtwo'$ incomplete trace,
    by the previous point we have that
    $\contone[\tdone]\smallstep{\trtwo'}\{(\contone_i[\tdone_i])^{\ppone_i}\}$
    and
    $(\contone,\tdone)\smallstep{\trtwo'}\{(\contone_i,\tdone_i)^{\ppone_i}\}$.
    So we have
    \begin{align*}
        \Pr(\contone[\tdone],\trtwo)=&
            \sum\ppone_i\cdot(\contone_i[\tdone_i]\smallstepreal{\view{\strp\strone}})=\sum\ppone_i\cdot
            \left.
            \begin{array}{ll}
                1, & \hbox{if $\contone_i=\strp\strone$;} \\
                0, & \hbox{if $\contone_i=\strp\strone'\neq\strp\strone$;} \\
                \tdone(\strp\strone), & \hbox{if $\contone=\hole$.} \\
            \end{array}%
            \right.=\\
        =&
            \sum\ppone_i\cdot((\contone_i,\tdone_i)\smallstepreal{\view{\strp\strone}})=(\contone,\tdone)\smallstepreal\trtwo
    \end{align*}
\end{itemize}
\end{proof}
But how could we exploit context pairs for our purposes? The
key idea can be informally explained as follows: there is
a notion of ``relatedness'' for pair distributions
which not only is stricter than trace equivalence, but
can be proved to be preserved along reduction, even when 
interaction with the environment is taken into account.
\newcommand{\ione}{i}
\newcommand{\isone}{I}
\begin{definition}[Trace Relatedness]
Let $\pdone,\pdtwo$ be two pair distributions. We say that they
are \emph{trace-related}, and we write $\pdone\sister\pdtwo$ if there
exist families $\{\contone_\ione\}_{\ione\in\isone}$,
$\{\tdone_\ione\}_{\ione\in\isone}$,
$\{\tdtwo_\ione\}_{\ione\in\isone}$, and
$\{\ppone_\ione\}_{\ione\in\isone}$ such that
$\pdone=\{(\contone_i,\tdone_i)^{\ppone_i}\}, \pdtwo=\{(\contone_i,\tdtwo_i)^{\ppone_i}\}$
and for every $\ione\in\isone$, it holds that $\tdone_i\treq\tdtwo_i$.
\end{definition}
%
%
The first observation about trace relatedness has to do with stability
with respect to internal reduction:
\begin{lemma}[Internal Stability]
Let $\pdone,\pdtwo$ be two pair distributions such that
$\pdone\sister\pdtwo$ then, if there exists $\pdone'$ such that
$\pdone\smallstep\emptyt\pdone'$, then there exists $\pdtwo'$ such
that $\pdtwo\smallstep\emptyt\pdtwo'$ and $\pdone'\sister\pdtwo'$.
\end{lemma}
\begin{proof}
By definition of $\sister$ for all
$(\contone,\tdone)^\ppone\in\pdone$ there exists
$(\contone,\tdtwo)^\ppone\in\pdtwo$ such that
$\tdone\treq\tdtwo$.\\
If $\pdone\smallstep\emptyt\pdone'$ then we have either $\pdone'=\pdone$ or $\pdone\pairconv\pdone'$; if $\pdone'=\pdone$ then we choose $\pdtwo'=\pdtwo$ and we get the thesis.\\
If $\pdone\pairconv\pdone'$ then we have that there exists a term
$(\contone,\tdone)\in\supp(\pdone)$ that reduces; we face two
possible cases:
\begin{itemize}
\item
    The first case is a term distribution reduction, i.e.
    $(\contone,\tdone)\onestep\{(\contone,\tdone')^1\}$.\\
    By the small step rules we know that
    $\pdone'=\pdone\setminus\{(\contone,\tdone)^\ppone\}+\{(\contone,\tdone')^\ppone\}$,
    but, given $(\contone,\tdtwo)^\ppone\in\pdtwo$ with $\tdone\treq\tdtwo$ by a previous lemma we know
    $\tdone'\treq\tdtwo$ and then if we set $\pdtwo'=\pdtwo$ we have the thesis.
\item
    The second case is a context reduction, i.e.
    $(\contone,\tdone)\onestep\{(\contone_i,\tdone_i)^{\ppone_i}\}$.\\
    We focus our attention on one particular reduction.\\
    Suppose that the pair that reduces is
    $(\ccase{\typeone}{\contone}{\tone_0}{\tone_1}{\tone_\epsilon},\tdone)^\ppone\in\pdone$, with $(\contone,\tdone)$ value; we know that there
    exists
    $(\ccase{\typeone}{\contone}{\tone_0}{\tone_1}{\tone_\epsilon},\tdtwo)^\ppone\in\pdtwo$
    such that $\tdone\treq\tdtwo$.\\
    If $\contone=\strp\strone$ by the one-step rules we have:
    $$(\ccase{\typeone}{\strp\strone}{\tone_0}{\tone_1}{\tone_\epsilon},\tdone)\onestep\left\{%
    \begin{array}{ll}
        \{(\tone_0,\tdone)^1\}, & \hbox{If $\strp\strone=\strp{\zero\strtwo}$;} \\
        \{(\tone_1,\tdone)^1\}, & \hbox{If $\strp\strone=\strp{\one\strtwo}$;} \\
        \{(\tone_\epsilon,\tdone)^1\}, & \hbox{If $\strp\strone=\strp\emptys$.} \\
    \end{array}%
    \right.$$
    and similarly:
    $$(\ccase{\typeone}{\strp\strone}{\tone_0}{\tone_1}{\tone_\epsilon},\tdtwo)\onestep\left\{%
    \begin{array}{ll}
        \{(\tone_0,\tdtwo)^1\}, & \hbox{If $\strp\strone=\strp{\zero\strtwo}$;} \\
        \{(\tone_1,\tdtwo)^1\}, & \hbox{If $\strp\strone=\strp{\one\strtwo}$;} \\
        \{(\tone_\epsilon,\tdtwo)^1\}, & \hbox{If $\strp\strone=\strp\emptys$.} \\
    \end{array}%
    \right.$$
    So we set:
    \begin{align*}
        \pdone'=\pdone\setminus\{(\ccase{\typeone}{\strp\strone}{\tone_0}{\tone_1}{\tone_\epsilon},\tdone)^\ppone\}+\{(\tone',\tdone)^\ppone\}\\
        \pdtwo'=\pdtwo\setminus\{(\ccase{\typeone}{\strp\strone}{\tone_0}{\tone_1}{\tone_\epsilon},\tdtwo)^\ppone\}+\{(\tone',\tdtwo)^\ppone\}\\
    \end{align*}
    where $\tone'$ is one between $\tone_0,\tone_1,\tone_\epsilon$ depending on $\strp\strone$, and we easily get the thesis.\\
    If $\contone=\hole$ then by the one-step rules we
    have:
    $$(\hole,\tdone)\onestep\{\tone_0^{\tdone(\strset_0)},\tone_1^{\tdone(\strset_1)},\tone_\epsilon^{\tdone(\strp\emptys)}\}\qquad
    (\hole,\tdtwo)\onestep\{\tone_0^{\tdtwo(\strset_0)},\tone_1^{\tdtwo(\strset_1)},\tone_\epsilon^{\tdtwo(\strp\emptys)}\}$$
    with
    $\strset_0=\{\strp{\zero\strone}\}_{\strp\strone\in\values^\typestring},\strset_1=\{\strp{\one\strone}\}_{\strp\strone\in\values^\typestring}$.\\
    But we know, for all $\strset$:
    \begin{align*}
        \tdone(\strset)=\sum_{\strp\strone\in\strset}\tdone(\strp\strone)=\sum_{\strp\strone\in\strset}\tdtwo(\strp\strone)=\tdtwo(\strset)
    \end{align*}
    So we have
    $$\pdone'=\pdone\setminus\{(\ccase{\typeone}{\hole}{\tone_0}{\tone_1}{\tone_\epsilon},\tdone)^\ppone\}
    +\{(\tone_0,\tdone)^{\ppone\cdot\tdone(\strset_0)},(\tone_1,\tdone)^{\ppone\cdot\tdone(\strset_1)},(\tone_\epsilon,\tdone)^{\ppone\cdot\tdone(\strp\emptys)}\}$$
    and if we set
    $$\pdtwo'=\pdtwo\setminus\{(\ccase{\typeone}{\hole}{\tone_0}{\tone_1}{\tone_\epsilon},\tdtwo)^\ppone\}
    +\{(\tone_0,\tdtwo)^{\ppone\cdot\tdtwo(\strset_0)},(\tone_1,\tdtwo)^{\ppone\cdot\tdtwo(\strset_1)},(\tone_\epsilon,\tdtwo)^{\ppone\cdot\tdtwo(\strp\emptys)}\}$$
    we obtain the thesis.\\
    The recursive case $\rrec{\typeone}{\contone}{\tone_0}{\tone_1}{\tone_\epsilon}$, with $(\contone,\tdone)$ value, is similar.\\
    On the other cases, if $(\contone,\tdone)\onestep\{(\contone_i,\tdone_i)^{\ppone_i}\}$
    by definition of $\sister$ we know that there must
    exist $(\contone,\tdtwo)^\ppone\in\pdtwo$ such that
    $(\contone,\tdtwo)\onestep\{(\contone_i,\tdtwo_i)^{\ppone_i}\}$
    (That is a reduction to the same contexts $\contone_i$ with the same probability
    $\ppone_i$), so we have to prove that $\tdone_i\treq\tdtwo_i$
    for all $i$. This is true because either the two term distributions
    remain the same, i.e. $\tdone_i=\tdone,\tdtwo_i=\tdtwo$ for
    all $i$., or the context passes the same value to the two term
    distributions and so by a previous lemma
    $\tdone_i\treq\tdtwo_i$ for all $i$.\\
    Now if we set
    $\pdone'=\pdone\setminus\{(\contone,\tdone)^\ppone+\{(\contone_i,\tdone_i)^{\ppone\cdot\ppone_i}\}\}$ and
    $\pdtwo'=\pdtwo\setminus\{(\contone,\tdtwo)^\ppone\}+\{(\contone_i,\tdtwo_i)^{\ppone\cdot\ppone_i}\}$
    then we have $\pdone\pairconv\pdone',\pdtwo\pairconv\pdtwo'$ and $\pdone'\sister\pdtwo'$.
\end{itemize}
\end{proof}
Once Internal Stability is proved, and since the relation $\pairconv$
can be proved to be strongly normalizing also for context pair
distributions, one gets that:
\begin{lemma}[Bisimulation, Internally]\label{lemma:bisimint}
If $\pdone,\pdtwo$ are pair distributions, with $\pdone\sister\pdtwo$
then there are $\pdone',\pdtwo'$ normal distributions such that
$\pdone\smallstep\emptyt\pdone'$, 
$\pdtwo\smallstep\emptyt\pdtwo'$ and
$\pdone'\sister\pdtwo'$.
\end{lemma}
\begin{proof}
The proof comes from the fact that, given $\pdone$ if it is not normal, there is $\pdone'$ normal such that $\pdone\pairconv^*\pdone'$, and by the previous lemma we have $\pdone'\sister\pdtwo$. Then if $\pdtwo$ isn't normal we can repeat the procedure and get $\pdtwo'$ such that $\pdtwo\pairconv^*\pdtwo'$ and $\pdone'\sister\pdtwo'$.
\end{proof}

The next step consists in proving that context pair distributions
which are trace related are not only bisimilar as for internal
reduction, but also for external reduction:
\begin{lemma}[Bisimulation, Externally]\label{lemma:bisimext}
Given two pair distributions $\pdone,\pdtwo$ with
$\pdone\sister\pdtwo$, then for all traces $\trtwo$ we have:
\begin{varenumerate}
\item 
  If $\pdone\smallstep\trtwo\pdthree$, with $\pdthree$ normal distribution, then $\pdtwo\smallstep\trtwo\pdfour$,
  where $\pdthree\sister\pdfour$ and $\pdfour$ is a normal distribution too.
\item 
  If $\pdone\smallstepreal\trtwo\ppone$ then $\pdtwo\smallstepreal\trtwo\ppone$.
\end{varenumerate}
\end{lemma}
\begin{proof}
We act by induction on the length of $\trtwo$.\\ If $\trtwo=\emptyt$
then by lemma~\ref{lemma:bisimint} we get the thesis.  Suppose now
$\trtwo=\trtwo'\cdot\pass\vone$ then we have by induction hypothesis:
$\pdone\smallstep{\trtwo'}\{(\contone_i,\tdone_i)^{\ppone_i}\}_{i\in I}$ and
$\pdtwo\smallstep{\trtwo'}\{(\contone_i,\tdtwo_i)^{\ppone_i}\}_{i\in I}$ with
$\tdone_i\treq\tdtwo_i$ for all $i\in I$ and the two pair distribution
normal.\\ But, by the one-step rules we have only two possible
derivation for an action $\pass\vone$:
$$
    \AxiomC{}
    \UnaryInfC{$(\lambda\varone.\contone,\tdone)\onestep^{\pass\vone}\{(\contone\subs{\vone}{\varone},\tdone)^1\}$}
    \DisplayProof
    \qquad
    \AxiomC{$\tdone\onestep^{\pass\vone}\tdone'$}
    \UnaryInfC{$(\hole,\tdone)\onestep^{\pass\vone}\{(\hole,\tdone')^1\}$}
    \DisplayProof
$$
So if we set $J=\{j\in\ I\ |\
\contone_j=\lambda\varone.\contone_j'\}, K=\{k\in\ I\ |\
\contone_i=\hole\}$ we have:
\begin{align*}
    \pdone\smallstep{\trtwo'}\{(\lambda\varone.\contone_j',\tdone_j)^{\ppone_j}\}+\{(\hole,\tdone_k)^{\ppone_k}\}\qquad
    \pdtwo\smallstep{\trtwo'}\{(\lambda\varone.\contone_j',\tdtwo_j)^{\ppone_j}\}+\{(\hole,\tdtwo_k)^{\ppone_k}\}
\end{align*}
At this point, if $\tdone_k\onestep^{\pass\vone}\tdone_k'$ and
$\tdtwo_k\onestep^{\pass\vone}\tdtwo_k'$ we know
$\tdone'_k\treq\tdtwo'_k$ for all $k$, so by using the one step rule
we set:
\begin{align*}
    \pdone'=\{(\contone_j\subs{\vone}{\varone},\tdone_j)^{\ppone_j}\}+\{(\hole,\tdone_k')^{\ppone_k}\}\qquad
    \pdtwo'=\{(\contone_j\subs{\vone}{\varone},\tdtwo_j)^{\ppone_j}\}+\{(\hole,\tdtwo_k')^{\ppone_k}\}
\end{align*}
and we have $\pdone\smallstep{\trtwo'\cdot\pass\vone}\pdone'$,
$\pdtwo\smallstep{\trtwo\cdot\pass\vone}\pdtwo'$ with
$\pdone''\sister\pdtwo''$; and so by applying lemma~\ref{lemma:bisimint} we get the (1)
thesis.\\\\
Suppose now $\trtwo=\trtwo'\cdot\view{\strp\strone}$.\\
By induction we know that
$\pdone\smallstep{\trtwo'}\{(\contone_i,\tdone_i)^{\ppone_i}\},\pdtwo\smallstep{\trtwo'}\{(\contone_i,\tdtwo_i)^{\ppone_i}\}$ with $\tdone_i\treq\tdtwo_i$ for all $i\in I$ and that the two pair distributions are normal.\\
So we set $J=\{j\in\ I\ |\ \contone_j=\strp{\strone_j}\},K=\{k\in I\ |\ \contone_k=\hole\}$ and know:
\begin{align*}
	\pdone\smallstep{\trtwo'}\{(\strp{\strone_j},\tdone_j)^{\ppone_j}\}+\{(\hole,\tdone_k)^{\ppone_k}\}
	\qquad
	\pdtwo\smallstep{\trtwo'}\{(\strp{\strone_j},\tdtwo_j)^{\ppone_j}\}+\{(\hole,\tdtwo_k)^{\ppone_k}\}
\end{align*}
So we have:
\begin{align*}
	\pdone\smallstepreal{\trtwo'\cdot\view\strp{\strone}}\sum_{\strp{\strone_j}=\strp\strone}\ppone_j + \sum\ppone_k\cdot\Pr(\tdone_k,\view{\strp\strone})
	\qquad
	\pdtwo\smallstepreal{\trtwo'\cdot\view{\strp\strone}}\sum_{\strp{\strone_j}=\strp{\strone}}\ppone_j + \sum\ppone_k\cdot\Pr(\tdtwo_k,\view{\strp\strone})
\end{align*}
But $\Pr(\tdone_k,\view{\strp\strone})=\Pr(\tdtwo_k,\view{\strp\strone})$ and so we get the thesis (2).
\end{proof}

%
\begin{lemma} 
Given two terms distributions $\tdone,\tdtwo$ such that
$\tdone\treq\tdtwo$, then for all context $\contone$, for all trace
$\trtwo$ we have:
$\Pr(\contone[\tdone],\trtwo)=\Pr(\contone[\tdtwo],\trtwo)$
\end{lemma}
\begin{proof}
If the trace $\trtwo$ doesn't end with the action $\view\cdot$ then
$\Pr(\contone[\tdone],\trtwo)=1=\Pr(\contone[\tdtwo],\trtwo)$.\\ 
Otherwise we know that
$(\contone,\tdone)\smallstepreal\trtwo\ppone$, we can write $\Pr((\contone,\tdone),\trtwo)=\ppone$, and by Lemma~\ref{lemma:bisimext}
we know $(\contone,\tdtwo)\smallstepreal\trtwo\ppone$. But
by Lemma~\ref{lemma:pairconv} we know
$\Pr(\contone[\tdone],\trtwo)=\Pr((\contone,\tdone),\trtwo)=\Pr((\contone,\tdtwo),\trtwo)=\Pr(\contone[\tdtwo],\trtwo)$ and
then the thesis.
\end{proof}
We are now in a position to prove the main result of this section:
\begin{theorem}
Trace equivalence is a congruence.
\end{theorem}
\begin{proof}
We have to prove that, given two terms $\tone,\ttwo$ such that
$\tone\treq\ttwo$ then for all contexts $\contone$, we have that
$\contone[\tone]\treq\contone[\ttwo]$, i.e., for all traces
$\trtwo$ we have
$\Pr(\contone[\tone],\trtwo)=\Pr(\contone[\ttwo],\trtwo)$.
But by Lemma~\ref{lemma:pairconv} and Lemma~\ref{lemma:bisimext} we have, indeed, that
$\Pr(\contone[\tone],\trtwo)=\Pr((\contone,\{\tone^1\}),\trtwo)=
\Pr((\contone,\{\ttwo^1\},\trtwo)=\Pr(\contone[\ttwo],\trtwo)$,
because the two pair distributions $\{(\contone,\{\tone^1\})^1\}$
and $\{((\contone,\{\ttwo^1\}))^1\}$ are trace-related.
\end{proof}
%
\begin{corollary}[Soundness]
Trace equivalence is included into context equivalence.
\end{corollary}
\begin{proof}
If $\tone\treq\ttwo$, then by the previous theorem we have that for
all contexts $\contone$ we have $\contone[\tone]\treq\contone[\ttwo]$
and this means that if we choose a trace $\trone=\view{\strp\emptys}$ then we
have
$\sem{\contone[\tone]}(\strp\emptys)=\Pr(\contone[\tone],\view{\strp\emptys})=\Pr(\contone[\ttwo],\view{\strp\emptys})=\sem{\contone[\ttwo]}(\strp\emptys)$,
and so the thesis.
\end{proof}
\begin{theorem}[Full Abstraction]
Context equivalence coincides with trace equivalence
\end{theorem}
\begin{proof}
For any admissible trace $\trone$ for $\typeone$, there is
a context $\contone_{\trone}\hole$ such that
$\Pr(\tone,\trone)=\sem{\contone_{\trone}[\tone]}(\strp\emptys)$,
which can be proved by induction on the structure of $\typeone$.
\end{proof}
\subsection{Some Words on Applicative Bisimulation}
As we already discussed, the quantification over all contexts makes
the task of proving two terms to be context equivalent burdensome,
even if we restrict to linear contexts. And we cannot say that trace
equivalence really overcomes this problem: there is a universal
quantification anyway, even if contexts are replaced by objects
(i.e. traces) having a simpler structure. It is thus natural to look
for other techniques. The interactive view provided by traces suggests
the possibility to go for coinductive techniques akin to Abramsky's
applicative bisimulation, which has already been shown to be adaptable
to probabilistic $\lambda$-calculi~\cite{DalLagoSangiorgiAlberti,DalLagoCrubille}.

First of all, we introduce a Labeled Transition System, by
defining a Labeled Markov chain
$\mcone=(\mcstates,\mclabels,\mcprob)$ where
$\mcstates=\terms\uplus\values$ is the set of states,
$\mclabels=\{\eval,\pass\cdot,\view\cdot\}$ is the set of labels, $\types$ is the set of types  and
$\mcprob$ is the probability measure defined as follows:
\begin{align*}
	\mcprob:(\mcstates,\types)\times\mclabels\times(\mcstates,\types)\to[0,1]\\
	\mcprob((\tone,\typeone),\eval,(\vone,\typeone))=\sem\tone(\vone) &
		\qquad
		\mcprob((\lambda\varone.\tone,\aone\typeone\to\typetwo),\pass\vone,(\tone\subs{\vone}{\varone},\typetwo))=1\\
	\mcprob((\strp\strone,\typestring),\view{\strp\strone},(\strp\strone,\typestring))=1 &
		\qquad
		\mcprob((\strp\strone,\typestring),\view{\strp\strone'},(\strp\strone,\typestring))=0
\end{align*}
So, before giving the definition of bisimulation we define a typed
relation as a family
$\relone=(\relone_\tcontone^\typeone)_{\typeone,\tcontone}$, where
each $\relone_\tcontone^\typeone$ is a binary relation on
$\terms_\tcontone^\typeone$; we define the open extension
$\open\relone$ by saying that, given two terms $\tone,\ttwo$ we have
$\tone\open\relone\ttwo$ iff for all $\tcontone$-closure $\xi$ we have
that $\proves(\tone\xi)\relone(\ttwo\xi):\typeone$.
\begin{definition}\label{def_pab}
Given a Labeled Markov Chain $\mcone=(\mcstates,\mclabels,\mcprob)$ a
\emph{probabilistic applicative bisimulation} is an equivalence relation
$\relone$ between the states of the Markov chain such that, given two
states $\tone,\ttwo$ we have $(\tone\ \relone\ \ttwo):\typeone$ if and
only if for each equivalence class $\eqcone$ modulo $\pabone$ we have:
$$\mcprob((\tone,\typeone),\labone,\eqcone)=\mcprob((\ttwo,\typeone),\labone,\eqcone)$$
We define $\sim$ as the reflexive and transitive closure of
$\bigcup\{\relone\ |\ \relone\ \text{bisimulation}\}$. We say that
two terms $\tone,\ttwo\in\terms_\typeone^\tcontone$ are bisimilar
(We write $\tcontone\proves\tone\sim\ttwo:\typeone$) if there
exists a bisimulation between them and we define $\open\sim$ as the bisimulation equivalence.
\end{definition}
\begin{definition}
A \emph{probabilistic applicative bisimulation} is defined
to be any type-indexed family of relations $\{\relone_\typeone\}_{\typeone\in\types}$
such that for each $\typeone$, $\relone_\typeone$ is an equivalence
relation over the set of closed terms of type $\typeone$, and
moreover the following holds:
\newcommand{\eqone}{E}
\newcommand{\subst}[3]{#1\{#2/#3\}}
\begin{varitemize}
\item
  If $\tone\relone_\typeone\ttwo$, then for every equivalence
  relation $\eqone$ modulo $\relone_\typeone$,
  it holds that $\sem{\tone}(\eqone)=\sem{\ttwo}(\eqone)$.
\item
  If $(\lambda\varone.\tone)\relone_{\aone\typeone\to\typetwo}(\lambda\varone.\ttwo)$,
  then for every closed value $\vone$ of type $\typeone$, it holds that
  $(\tone\subs{\vone}{\varone})\relone_{\typetwo}(\ttwo\subs{\vone}{\varone})$.
\item
  If $\strp\strone\relone_\typestring\strp\strtwo$, then $\strp\strone=\strp\strtwo$.
\end{varitemize}
\end{definition}
With some effort, one can prove that a greatest applicative
bisimulation exists, and that it consists of the union (at any type)
of all bisimulation relations. This is denoted as $\sim$ and
said to be (applicative) \emph{bisimilarity}. One can then
generalize $\sim$ to a relation $\open\sim$ on open terms by the usual
open extension.

One way to show that bisimilarity is included in context equivalence
consists in proving that $\open\sim$ is a congruence; to reach this
goal we first lift $\open\sim$ to another relation $\hopen\sim$ by the
so-called \emph{Howe's method}~\cite{Howe96}, and then transitive
close it, obtaining another relation $\hotc\sim$. This can be done by
the rules in Figure~\ref{figure:howel}. By construction, the relation
$\hotc\sim$ is a congruence. But one can also show that it coincides
with $\open\sim$, namely that $\hotc\sim\subseteq\open\sim$ and
$\open\sim\subseteq\hotc\sim$.
\begin{figure}[!h]
\centering \fbox{
\begin{minipage}{.95\textwidth}
\footnotesize
{
	$$
	\AxiomC{$\strp\strone\open\sim \tone$}
	\RightLabel{$\How0$}
	\UnaryInfC{$\strp\strone\hopen\sim \tone$}
	\DisplayProof
	\qquad
	\AxiomC{$\varone\open\sim \tone$}
	\RightLabel{$\How1$}
	\UnaryInfC{$\varone\hopen\sim\tone$}
	\DisplayProof
	\qquad
	\AxiomC{$\begin{array}{cc}
		\tone\free\varone\hopen\sim\ttwo\free\varone & \lambda\varone.\ttwo\sim\tthree
	\end{array}$}
	\RightLabel{$\How2$}
	\UnaryInfC{$\lambda\varone.\tone\hopen\sim\tthree$}
	\DisplayProof
	\qquad
	\AxiomC{$\begin{array}{cc}
		\tone_1\hopen\sim\ttwo_1 & \\
		\tone_2\hopen\sim\ttwo_2 & \ttwo_1\ttwo_2\sim\tthree
	\end{array}$}
	\RightLabel{$\How3$}
	\UnaryInfC{$\tone_1\tone_2\hopen\sim\tthree$}
	\DisplayProof
	$$
	\vspace{6pt}
	$$
	\AxiomC{$\begin{array}{cc}
		\tone\hopen\sim\ttwo &\\
		\tone_0\hopen\sim\ttwo_0 & \tone_\epsilon\hopen\sim\ttwo_\epsilon\\
		\tone_1\hopen\sim\ttwo_1 & \ccase{\typeone}{\ttwo}{\ttwo_0}{\ttwo_1}{\ttwo_\epsilon}\sim\tthree
		\end{array}$}
	\RightLabel{$\How4$}
	\UnaryInfC{$\ccase{\typeone}{\tone}{\tone_0}{\tone_1}{\tone_\epsilon}\hopen\sim\tthree$}
	\DisplayProof
	\qquad
	\AxiomC{$\begin{array}{cc}
		\tone\hopen\sim\ttwo & \\
		\tone_0\hopen\sim\ttwo_0 & \tone_\epsilon\hopen\sim\ttwo_\epsilon\\
		\tone_1\hopen\sim\ttwo_1 & \rrec{\typeone}{\ttwo}{\ttwo_0}{\ttwo_1}{\ttwo_\epsilon}\sim\tthree
	\end{array}$}
	\RightLabel{$\How5$}
	\UnaryInfC{$\rrec{\typeone}{\tone}{\tone_0}{\tone_1}{\tone_\epsilon}\hopen\sim\tthree$}
	\DisplayProof
	$$
	\vspace{6pt}
	$$
	\AxiomC{$\tone\hopen\sim\ttwo$}
	\RightLabel{$\TC1$}
	\UnaryInfC{$\tone\hotc\sim\ttwo$}
	\DisplayProof
	\qquad
	\AxiomC{$\tone\hotc\ttwo$}
	\AxiomC{$\ttwo\hotc\tthree$}
	\RightLabel{$\TC2$}
	\BinaryInfC{$\tone\hotc\tthree$}
	\DisplayProof
	$$
}
\end{minipage}}
\caption{Howe's Lifting and Transitive Closure Rules}\label{figure:howel}
\end{figure}
The first inclusion is again an easy consequence of the way
$\hotc\sim$ is defined, and of the fact that $\sim$ is an equivalence
relation. The second one is more difficult, and needs some
intermediary steps to get proved. The first step is given by the
following lemma.
\begin{lemma}[Key Lemma]
Given two terms $\tone,\ttwo$, we have:
\begin{varitemize}
\item 
  If $\proves\tone\hopen\sim\ttwo:\aone\typeone\to\typetwo$, then for
  all
  $\eqcone\in\terms_{\varone:\aone\typeone}^{\typetwo}/_{\hopen\sim}$
  equivalence class modulo $\hopen\sim$ it holds that
  $\sem\tone(\lambda\varone.\eqcone)=\sem\ttwo(\lambda\varone.\eqcone)$.
\item 
  If $\proves\tone\hopen\sim\ttwo:\typestring$, then for all
  $\strp\strone\in\values^\typestring$ we have
  $\sem\tone(\strp\strone)=\sem\ttwo(\strp\strone)$.
\end{varitemize}
\end{lemma}
\begin{proof}
We work by induction on the derivation of $\proves\tone\hopen\sim\ttwo$.
\begin{itemize}
\item
	Suppose $\tone=\strp\strone$, then we have $\proves\strp\strone\hopen\sim\ttwo:\typestring$ that is derived from $\How0$:
	$$\AxiomC{$\proves\strp\strone\open\sim\ttwo$}
	    \UnaryInfC{$\proves\strp\strone\hopen\sim\ttwo$}
	    \DisplayProof
	$$
	So we have, for all $\strp\strone'\in\values^\typestring$, by definition of $\open\sim$:
	$$\sem\strp\strone(\strp\strone')=\sem\ttwo(\strp\strone')$$
	
\item
	Suppose $\tone=\lambda\varone.\tone'$ then we have $\proves\lambda\varone.\tone'\hopen\sim\ttwo:\aone\typeone\to\typetwo$, derived from $\How2$:
	$$
	\AxiomC{$\begin{array}{cc}
		\tone'\free\varone\hopen\sim\ttwo'\free\varone & \lambda\varone.\ttwo'\open\sim\ttwo
	\end{array}$}
	\UnaryInfC{$\lambda\varone.\tone'\hopen\sim\ttwo$}
	\DisplayProof
	$$
	So, given $\eqcone\in\terms_{\varone:\aone\typeone}^{\typetwo}/_{\hopen\sim}$ we have:
	\begin{align*}
		\sem{\lambda\varone.\tone'}(\lambda\varone.\eqcone)=\left\{%
		\begin{array}{ll}
			1, & \hbox{if $\tone'\in\eqcone$;} \\
			0, & \hbox{otherwise.} \\
		\end{array}\right.=
		\left\{%
		\begin{array}{ll}
			1, & \hbox{if $\ttwo'\in\eqcone$;} \\
			0, & \hbox{otherwise.} \\
		\end{array}\right.=\sem{\lambda\varone.\ttwo'}(\lambda\varone.\eqcone)=\sem\ttwo(\lambda\varone.\eqcone)
	\end{align*}
\item 
	Suppose now $\tone=\ccase{\aone\typeone\to\typetwo}{\tone'}{\tone'_0}{\tone'_1}{\tone'_\epsilon}$, then: $\proves\ccase{\aone\typeone\to\typetwo}{\tone'}{\tone'_0}{\tone'_1}{\tone'_\epsilon}\hopen\sim\ttwo:\aone\typeone\to\typetwo$, which is derived from $\How4$:
	$$
	\AxiomC{$\begin{array}{cc}
		\tone'\hopen\sim\ttwo' &\\
		\tone_0'\hopen\sim\ttwo_0' & \tone_\epsilon'\hopen\sim\ttwo_\epsilon'\\
		\tone_1'\hopen\sim\ttwo_1' & \ccase{\aone\typeone\to\typetwo}{\ttwo'}{\ttwo_0'}{\ttwo_1''}{\ttwo_\epsilon}\sim\ttwo
		\end{array}$}
	\UnaryInfC{$\ccase{\aone\typeone\to\typetwo}{\tone'}{\tone_0'}{\tone_1'}{\tone_\epsilon'}\hopen\sim\ttwo$}
	\DisplayProof
	$$
	Then for all $\eqcone\in\terms_{\varone:\aone\typeone}^{\typetwo}/_{\hopen\sim}$ we have by induction hypothesis:
	\begin{align*}
		\sem\tone(\lambda\varone.\eqcone)=&
			\sem{\ccase{\aone\typeone\to\typetwo}{\tone'}{\tone'_0}{\tone'_1}{\tone'_\epsilon}}(\lambda\varone.\eqcone)=\\
		=&
			\sem{\tone'}(\strp\emptys)\sem{\tone'_\epsilon}(\lambda\varone.\eqcone) +\sum_{\strp\strone}\sem{\tone'}(\strp{\zero\strone})\sem{\tone'_0}(\eqcone)+\sum_{\strp\strone}\sem{\tone'}(\strp{\one\strone})\sem{\tone'_1}(\eqcone)=\\
		=&
			\sem{\ttwo'}(\strp\emptys)\sem{\ttwo'_\epsilon}(\lambda\varone.\eqcone) +\sum_{\strp\strone}\sem{\ttwo'}(\strp{\zero\strone})\sem{\ttwo'_0}(\eqcone)+\sum_{\strp\strone}\sem{\ttwo'}(\strp{\one\strone})\sem{\ttwo'_1}(\eqcone)=\sem\ttwo(\lambda\varone.\eqcone)
	\end{align*}
\item
	If $\tone=\ccase{\typestring}{\tone'}{\tone'_0}{\tone'_1}{\tone'_\epsilon}:\typestring$ the proof is similar to the previous case.
\item
	Suppose now $\tone=\tone_1\tone_2:\aone\typeone\to\typetwo$ and so we have $\proves\tone_1\tone_2\hopen\sim\ttwo:\aone\typeone\to \typetwo$ that is derived from:
	$$
	\AxiomC{$\begin{array}{cc}
		\tone_1\hopen\sim\ttwo_1 & \\
		\tone_2\hopen\sim\ttwo_2 & \ttwo_1\ttwo_2\sim\tthree
	\end{array}$}
	\UnaryInfC{$\tone_1\tone_2\hopen\sim\tthree$}
	\DisplayProof
	$$
	We have to face two different cases: $\tone_2\in\terms^\typestring$ and $\tone_2\in\terms^{\athree\typethree\to\typefour}$.	If $\tone_2\in\terms^\typestring$ then for all $\eqcone\in\terms_{\varone:\aone\typeone}^{\typetwo}/_{\hopen\sim}$ we have:
	\begin{align*}
		\sem\tone(\lambda\varone.\eqcone)=&
			\sem{\tone_1\tone_2}(\lambda\varone.\eqcone)=\sum_{\strp\strone'\in\values^\typestring}\sem{\tone_2}(\strp\strone')\left(\sum_{\eqcone_\tthree\in\values_{\vartwo:\astring\typestring}^{\aone\typeone\to\typetwo}}\sum_{\tthree\in\eqcone_\tthree}\sem{\tone_1}(\lambda\vartwo.\tthree)\sem{\tthree\subs{\strp\strone'}{\vartwo}}(\lambda\varone.\eqcone)\right)=\\
			=&
			\sum_{\strp\strone'\in\values^\typestring}\sem{\tone_2}(\strp\strone')\left(\sum_{\eqcone_\tthree\in\values_{\vartwo:\astring\typestring}^{\aone\typeone\to\typetwo}}\sem{\tone_1}(\lambda\vartwo.\eqcone_\tthree)\sem{\eqcone_\tthree\subs{\strp\strone'}{\vartwo}}(\lambda\varone.\eqcone)\right)=\\
			=&
			\sum_{\strp\strone'\in\values^\typestring}\sem{\ttwo_2}(\strp\strone')\left(\sum_{\eqcone_\tthree\in\values_{\vartwo:\astring\typestring}^{\aone\typeone\to\typetwo}}\sem{\ttwo_1}(\lambda\vartwo.\eqcone_\tthree)\sem{\eqcone_\tthree\subs{\strp\strone'}{\vartwo}}(\lambda\varone.\eqcone)\right)=\sem{\ttwo_1\ttwo_2}(\lambda\varone.\eqcone)=\\
			=&
			\sem\ttwo(\lambda\varone.\eqcone)
	\end{align*}
	If $\tone_2\in\terms^{\athree\typethree\to\typefour}$ then we have:
	\begin{align*}
		\sem\tone(\lambda\varone.\eqcone)=&
			\ \sem{\tone_1\tone_2}(\lambda\varone.\eqcone)=\\
		=&
			\sum_{\eqcone_\vone\in\values_{\varthree:\athree\typethree}^{\typefour}}\sum_{\vone\in\eqcone_\vone}\sem{\tone_2}(\lambda\varthree.\vone)
			\left(\sum_{\eqcone_\tthree\in\values_{\vartwo:\athree\typethree\to\typefour}^{\aone\typeone\to\typetwo}}\sum_{\tthree\in\eqcone_\tthree}\sem{\tone_1}(\lambda\vartwo.\tthree)\sem{\tthree\subs{\lambda\varthree.\vone}{\vartwo}}(\lambda\varone.\eqcone)\right)=\\
		=&
			\sum_{\eqcone_\vone\in\values_{\varthree:\athree\typethree}^{\typefour}}\sem{\tone_2}(\lambda\varthree.\eqcone_\vone)
			\left(\sum_{\eqcone_\tthree\in\values_{\vartwo:\athree\typethree\to\typefour}^{\aone\typeone\to\typetwo}}\sem{\tone_1}(\lambda\vartwo.\eqcone_\tthree)\sem{\eqcone_\tthree\subs{\lambda\varthree.\eqcone_\vone}{\vartwo}}(\lambda\varone.\eqcone)\right)=\\
		=&
			\sum_{\eqcone_\vone\in\values_{\varthree:\athree\typethree}^{\typefour}}\sem{\ttwo_2}(\lambda\varthree.\eqcone_\vone)
			\left(\sum_{\eqcone_\tthree\in\values_{\vartwo:\athree\typethree\to\typefour}^{\aone\typeone\to\typetwo}}\sem{\ttwo_1}(\lambda\vartwo.\eqcone_\tthree)\sem{\eqcone_\tthree\subs{\lambda\varthree.\eqcone_\vone}{\vartwo}}(\lambda\varone.\eqcone)\right)=\\
		&=
		\sem{\ttwo_1\ttwo_2}(\lambda\varone.\eqcone)
		= \sem\ttwo(\lambda\varone.\eqcone)
	\end{align*}
\item The case $\tone=\tone_1\tone_2:\typestring$ is similar to the previous one.
\item
	Finally, if $\tone=\rrec{\aone\typeone\to\typetwo}{\tone'}{\tone'_0}{\tone'_1}{\tone'_\epsilon}$ then we have $\proves\rrec{\aone\typeone\to\typetwo}{\tone'}{\tone'_0}{\tone'_1}{\tone'_\epsilon}\hopen\sim\ttwo$ which is derived from:
	$$
	\AxiomC{$\begin{array}{cc}
		\tone'\hopen\sim\ttwo' & \\
		\tone_0'\hopen\sim\ttwo_0 & \tone_\epsilon'\hopen\sim\ttwo_\epsilon'\\
		\tone_1'\hopen\sim\ttwo_1 & \rrec{\aone\typeone\to\typetwo}{\ttwo'}{\ttwo_0'}{\ttwo_1'}{\ttwo_\epsilon'}\sim\ttwo
	\end{array}$}
	\UnaryInfC{$\rrec{\aone\typeone\to\typetwo}{\tone'}{\tone_0'}{\tone_1'}{\tone_\epsilon'}\hopen\sim\ttwo$}
	\DisplayProof
	$$
	then for all $\eqcone\in\values_{\varone:\aone}^\typetwo/_{\hopen\sim}$ we have:
	\begin{align*}
		\sem\tone(\lambda\varone.\eqcone)&=
			\sem{\rrec{\aone\typeone\to\typetwo}{\tone'}{\tone'_0}{\tone'_1}{\tone'_\epsilon}}(\lambda\varone.\eqcone)=\\
			&=
			\sem{\tone'}(\strp\emptys)\sem{\tone_\epsilon}(\lambda\varone.\eqcone)+
			\sum_{\strp\strone\in\values^\typestring}\sem{\tone'}(\strp{\zero\strone})\sem{(\tone'_0\strp{\zero\strone})(\rrec{\aone\typeone\to\typetwo}{\strp\strone}{\tone_0'}{\tone_1'}{\tone_\epsilon})}(\lambda\varone.\eqcone)\\
			&\qquad + \sum_{\strp\strone\in\values^\typestring}\sem{\tone'}(\strp{\one\strone})\sem{(\tone'_1\strp{\one\strone})(\rrec{\aone\typeone\to\typetwo}{\strp\strone}{\tone_1'}{\tone_1'}{\tone_\epsilon})}(\lambda\varone.\eqcone)=\\
			&=
			\sem{\ttwo'}(\strp\emptys)\sem{\ttwo_\epsilon}(\lambda\varone.\eqcone)+
			\sum_{\strp\strone\in\values^\typestring}\sem{\ttwo'}(\strp{\zero\strone})\sem{(\ttwo'_0\strp{\zero\strone})(\rrec{\aone\typeone\to\typetwo}{\strp\strone}{\ttwo_0'}{\ttwo_1'}{\ttwo_\epsilon})}(\lambda\varone.\eqcone)\\
			&\qquad + \sum_{\strp\strone\in\values^\typestring}\sem{\ttwo'}(\strp{\one\strone})\sem{(\ttwo'_1\strp{\one\strone})(\rrec{\aone\typeone\to\typetwo}{\strp\strone}{\ttwo_1'}{\ttwo_1'}{\ttwo_\epsilon})}(\lambda\varone.\eqcone) = \sem\ttwo(\lambda\varone.\eqcone)
	\end{align*}
\item
	The case $\tone=\rrec{\typestring}{\tone'}{\tone'_0}{\tone'_1}{\tone'_\epsilon}$ is similar to the previous one.
\end{itemize}
So we have the thesis.\\
\end{proof}

\begin{theorem}
$\hotc\sim$ is a bisimulation.
\end{theorem}
\begin{proof} We work by induction on the derivation of $\hotc\sim$,
proving that $\hotc\sim$ satisfies the three points of the
definition above. This, in particular, relies on the Key Lemma.
\end{proof}
\begin{theorem}
Bisimilarity is a congruence.
\end{theorem}
\begin{proof}
  The proof comes easily from the fact that $\hotc\sim$ is a
  congruence. Indeed it is transitive and symmetric by definition and
  also compatible. By the definition of $\hotc\sim$ we have that
  $\open\sim\subseteq\hopen\sim\subseteq\hotc\sim$, but the theorem
  above tells us that $\hotc\sim$ is a bisimulation, and that it must
  be included in $\open\sim$, the symmetric and transitive closure of
  all the bisimulations. So we have
  $\open\sim\subseteq\hotc\sim\ \wedge\ \hotc\sim\subseteq\open\sim$
  which means that $\open\sim=\hotc\sim$, and we get the thesis,
  namely that $\open\sim$ is a congruence.
\end{proof}
As usual, being a congruence has soundness as an easy corollary:
\begin{corollary}[Soundness]
Bisimilarity is included in context equivalence.
\end{corollary}
Is there any hope to get full abstraction? The answer is negative:
applicative bisimilarity is too strong to match context equivalence.
A counterexample to that can be built easily following the analogous
one from \cite{DalLagoSangiorgiAlberti}. Consider the following two
terms:
\newcommand{\ifte}[3]{\mathtt{if}\;#1\;\mathtt{then}\;#2\;\mathtt{else}\;#3}
\newcommand{\ttrue}{\mathtt{true}}
\newcommand{\tfalse}{\mathtt{false}}
$$
\tone=\lambda\varone.\ifte{\rand}{\ttrue}{\tfalse};\qquad\ttwo=\ifte{\rand}{(\lambda\varone.\ttrue)}{(\lambda\varone.\tfalse)};
$$
where we have used some easy syntactic sugar. It is easy to show that
$\tone$ and $\ttwo$ are trace equivalent, thus context equivalent.
On the other hand, $\tone$ and $\ttwo$ cannot be bisimilar.

This, however, is not the end of the story on coinductive
methodologies for context equivalence in \RSLR. A different route,
suggested by trace equivalence, consists in taking the naturally
definable (deterministic) labeled transition system of
term \emph{distributions} and ordinary bisimilarity over it. What one
obtains this way is a precise characterization of context
equivalence. There is a price to pay however, since one is forced
to reason on distributions rather than terms.

\section{From Equivalences to Metrics}\label{sect:metrics}
The notion of observation on top of which context equivalence is
defined is the probability of evaluating to the empty string, and is
thus quantitative in nature. This suggests the possibility of
generalizing context equivalence into a notion of \emph{distance}
between terms:
\begin{definition}[Context Distance]
For every type $\typeone$, we define
$\conmet_\typeone:\terms^{\typeone}_{\emtcont}\times\terms^{\typeone}_{\emtcont}\to\RR_{[0,1]}$
as
$\conmet_\typeone(\tone,\ttwo)=\sup_{\proves\contone[\proves\typeone]:\typestring}\myabs{\sem{\contone[\tone]}(\strp\emptys)-\sem{\contone[\ttwo]}(\strp\emptys)}$.
\end{definition}
For every type $\typeone$, the function $\conmet_\typeone$ is
a \emph{pseudometric}\footnote{ Following the literature on the
subject, this stands for any function $\delta:A\times A\rightarrow\RR$
such that $\delta(x,y)=\delta(y,x)$, $\delta(x,x)=0$ and
$\delta(x,y)+\delta(y,z)\geq\delta(x,z)$} on the space of closed
terms. Obviously, $\conmet_\typeone(\tone,\ttwo)=0$ iff $\tone$ and $\ttwo$ are
context equivalent. As such, then, the context distance can be seen as
a natural generalization of context equivalence, where a real number
between $0$ and $1$ is assigned to each pair of terms and is meant to
be a measure of how different the two terms are in terms of their
behavior. $\conmet$ refers to the family $\{\conmet_\typeone\}_{\typeone\in\types}$.

One may wonder whether $\conmet$, as we have defined it, can 
somehow be characterized by a trace-based notion of metric, similarly to
what have been done in Section~\ref{sect:equivalences} for
equivalences. First of all, let us \emph{define} such a
distance. Actually, the very notion of a trace needs to be slightly
modified: in the action $\view{\cdot}$, instead of observing
a \emph{single} string $\strp\strone$, we need to be able to observe the
action on a \emph{finite string set} $\strset$. The probability of accepting
a trace in a term will be modified accordingly:
$\Pr(\tone,\view\strset)=\sem{\tone}(\strset)$.
\begin{definition}[Trace Distance]
For every type $\typeone$, we define
$\tramet_\typeone:\terms^{\typeone}_{\emtcont}\times\terms^{\typeone}_{\emtcont}\to\RR_{[0,1]}$ as
$\tramet_\typeone(\tone,\ttwo)=\sup_{\trone}\myabs{\Pr(\tone,\trone)-\Pr(\ttwo,\trone)}$. 
\end{definition} 
It is easy to realize that if $\tone\treq\ttwo$ then $\tramet_\typeone(\tone,\ttwo)=0$.
Moreover, $\tramet_\typeone$ is itself a pseudometric. As usual, $\tramet$ denotes
the family $\{\tramet_\typeone\}_{\typeone\in\types}$.

But how should we proceed if we want to prove the two just introduced
notions of distance to coincide? Could we proceed more or less like
in Section~\ref{sect:traceequivalence}? The answer is positive, but
of course something can be found which plays the role of compatibility,
since the latter is a property of \emph{equivalences} and not of \emph{metrics}.
The way out is relatively simple: what corresponds to compatibility in
metrics is non-expansiveness (see, e.g., \cite{Desharnais99}). A notion of distance $\delta$ is
said to be \emph{non-expansive} iff for every pair of terms $\tone,\ttwo$
and for every context $\contone$, it holds that 
$\delta(\contone[\tone],\contone[\ttwo])\leq\delta(\tone,\ttwo)$,
that is a pseudometric too.\\
Now we show some properties of the trace distance $\tramet$ applied on term distributions.
\begin{lemma}\label{lemma:trametprop} \textsf{\textbf{(Trace Distance Properties):}}
Given two term distributions $\tdone,\tdtwo$ such that $\tramet(\tdone,\tdtwo)=\distone$ then we have:
\begin{enumerate}
\item
    If $\tdone\pairconv\tdone'$ then $\tramet(\tdone',\tdtwo)=\distone$.
\item
    If $\tdone\smallstep{\pass\vone}\tdone',\ \tdtwo\smallstep{\pass\vone}\tdtwo'$ then $\tramet(\tdone',\tdtwo')\leq\distone$.
\item
    If $\tdone\smallstepreal{\view\strset}\ppone_1,\ \tdtwo\smallstepreal{\view\strset}\ppone_2$ then $|\ \ppone_1-\ppone_2\ |\leq\distone$.
\end{enumerate}
\end{lemma}
\begin{proof}
\begin{enumerate}
\item
    Suppose $\tdone=\{(\tone_i)^{\ppone_i}\}$, then we have that
    there exists
    $\supp_\tdone\ni\tone'\onestep\{(\tone'_j)^{\ppone'_j}\}$; we set
    $\tdone'=\tdone\setminus\{(\tone')^\ppone\}+\{(\tone'_j)^{\ppone\cdot\ppone'_j}\}$. By the small step rules we have $\tdone\smallstep\emptyt\tdone'$, so we have for every trace $\trtwo$:
    \begin{align*}
    \Pr(\tdone,\trtwo) = Pr(\tdone,\emptyt\cdot\trtwo) = Pr(\tdone',\trtwo)
    \end{align*}
    So if for all traces $\trtwo$, $\Pr(\tdone,\trtwo)=\Pr(\tdone',\trtwo)$ then we have the thesis,
    $\tramet(\tdone,\tdtwo)=\tramet(\tdone',\tdtwo)$.
\item
    It comes from the fact the the quantification is over a
    smaller set of traces, so the distance can't be greater.
\item
    It comes from the fact that the quantification catches the
    trace $\view\strset$.
\end{enumerate}
\end{proof}
\begin{definition}
Given two pair distributions $\pdone,\pdtwo$, we say that they are
\emph{$\distone$-related}, we write
$\pdone\relative\distone\pdtwo$, if there exist
$\{\contone_i\}_{i\in I}$ contexts, $\{\tdone_i\}_{i\in
I},\{\tdtwo_i\}_{i\in I}$ term distributions,
$\{\ppone_i\},\{\pptwo_i\},\{\ppthree_i\}$ probabilities with
$\sum\ppone_i=\sum\pptwo_i=\sum\ppthree_i=1$ such that: $\pdone=\{(\contone_i,\tdone_i)^{\ppone_i}\},\ \pdtwo=\{(\contone_i,\tdtwo_i)^{\pptwo_i}\}$, with:
\begin{align*}
    \tramet(\tdone_i,\tdtwo_i)\leq\distone \text{ for all } i
    \qquad
    \wedge
    \qquad
    \left\{%
    \begin{array}{ll}
        \ppone_i=\pptwo_i=\ppthree_i, & \hbox{If $\contone_i\neq\tone$;} \\
        |\ppone_i-\pptwo_i|\leq\ppthree_i\cdot\distone, & \hbox{If $\contone_i=\tone_i$.} \\
    \end{array}%
    \right.
\end{align*}
\end{definition}
\begin{lemma} [Internal $d$-stability]\label{lemma:intdstab}
Given two pair distributions $\pdone,\pdtwo$ with $\pdone\relative\distone\pdtwo$ then if there exists $\pdone'$ such that $\pdone\pairconv\pdone'$ then there exists $\pdtwo'$
such that $\pdtwo\smallstep\emptyt\pdtwo'$ or $\pdtwo\pairconv\pdtwo'$ and $\pdone'\relative\distone\pdtwo'$.
\end{lemma}
\begin{proof}
The pair distribution $\pdone$ can reduce to $\pdone'$ in
two different ways: we could have a term distribution reduction,
i.e.
$\supp_\pdone\ni(\contone,\tdone)\onestep\{(\contone,\tdone')^1\}$,
or a context reduction, i.e.
$\supp_\pdone\ni(\contone,\tdone)\onestep\{(\contone_i,\tdone_i)^{\ppone_i}\}$,
so, let's prove the statement for the two cases:
\begin{enumerate}
\item
    Term distribution reduction.\\
    Suppose that the term in $\pdone$ that reduces is
    $(\contone,\tdone)^\ppone$;
    by definition there exists
    $\pdtwo\ni(\contone,\tdtwo)^\pptwo$ such that
    $\tramet(\tdone,\tdtwo)\leq\distone$ and
    $\ppone=\pptwo=\ppthree$ if $\contone\neq\tone$, $|\ppone-\pptwo|\leq\ppthree\cdot
    \distone$ otherwise.\\
    If $\tdone\pairconv\tdone'$, by the small step rules we have
    $\pdone\pairconv\pdone'=\pdone\setminus\{(\contone,\tdone)^\ppone\}+\{(\contone,\tdone')^\ppone\}$;
    so if we set $\pdtwo'=\pdtwo$ we have
    $\pdtwo\smallstep\emptyt\pdtwo'$ and obviously
    $\pdone'\relative\distone\pdtwo'$.
\item
    Context reduction\\
    When we face a context reduction we have to work on different
    cases:
    \begin{enumerate}
    \item
        Suppose that the pair that reduces is
        $\pdone\ni(\ccase{\typeone}{\contone}{\tone_0}{\tone_1}{\tone_\epsilon},\tdone)^\ppone$ with $\contone\in\values^\typestring$.\\
        If $\contone=\strp\strone$ then there exists $\pdtwo\ni(\ccase{\typeone}{\strp\strone}{\tone_0}{\tone_1}{\tone_\epsilon},\tdtwo)^\pptwo$
        with $\tramet(\tdone,\tdtwo)\leq\distone$ and $|\ppone-\pptwo|\leq\ppthree\cdot\distone$.\\
        By the one-step rules we have:
        $$(\ccase{\typeone}{\strp\strone}{\tone_0}{\tone_1}{\tone_\epsilon},\tdone)\onestep\left\{%
        \begin{array}{ll}
            \{(\tone_0,\tdone)^1\}, & \hbox{If $\strp\strone=\strp{\zero \strtwo}$;} \\
            \{(\tone_0,\tdone)^1\}, & \hbox{If $\strp\strone=\strp{\one \strtwo}$;} \\
            \{(\tone_\epsilon,\tdone)^1\}, & \hbox{If $\strp\strone=\strp\emptys$.} \\
        \end{array}%
        \right.$$
        and the same for
        $$(\ccase{\typeone}{\strp\strone}{\tone_0}{\tone_1}{\tone_\epsilon},\tdtwo)\onestep\left\{%
        \begin{array}{ll}
            \{(\tone_0,\tdtwo)^1\}, & \hbox{If $\strp\strone=\strp{\zero \strtwo}$;} \\
            \{(\tone_1,\tdtwo)^1\}, & \hbox{If $\strp\strone=\strp{\one\strtwo}$;} \\
            \{(\tone_\epsilon,\tdtwo)^1\}, & \hbox{If $\strp\strone=\strp\emptys$.} \\
        \end{array}%
        \right.$$
        So we have
        $$\pdone'=\pdone\setminus\{(\ccase{\typeone}{\strp\strone}{\tone_0}{\tone_1}{\tone_\epsilon},\tdone)^\ppone\}+\{(\tone',\tdone)^\ppone\}$$
        and if we set
        $$\pdtwo'=\pdtwo\setminus\{(\ccase{\typeone}{\strp\strone}{\tone_0}{\tone_1}{\tone_\epsilon},\tdtwo)^\pptwo\}+\{(\tone',\tdtwo)^\pptwo\}$$
        where $\tone'$ is one between $\tone_0,\tone_1,\tone_\epsilon$ depending on
        $\strp\strone$, then we get the thesis.\\\\
        If $\contone=\hole$ there exists $\pdtwo\ni(\ccase{\typeone}{\hole}{\tone_0}{\tone_1}{\tone_\epsilon},\tdtwo)^\pptwo$ such
        that $\tramet(\tdone,\tdtwo)\leq\distone$ and
        $\ppone=\pptwo=\ppthree$.\\
        By the one step rules we
        have:
        $$(\hole,\tdone)\onestep\{\tone_0^{\tdone(\strset_0)},\tone_1^{\tdone(\strset_1)},\tone_\epsilon^{\tdone(\strp\emptys)}\}\qquad
        (\hole,\tdtwo)\onestep\{\tone_0^{\tdtwo(\strset_0)},\tone_1^{\tdtwo(\strset_1)},\tone_\epsilon^{\tdtwo(\strp\emptys)}\}$$
        with
        $\strset_0=\{\strp{\zero\strone}\}_{\strp\strone\in\values^\typestring},\strset_1=\{\strp{\one\strone}\}_{\strp\strone\in\values^\typestring}$.\\
        So we have
        $$\pdone'=\pdone\setminus\{(\ccase{\typeone}{\hole}{\tone_0}{\tone_1}{\tone_\epsilon},\tdone)^\ppone\}
        +\{(\tone_0,\tdone)^{\ppone\cdot\tdone(\strset_0)},(\tone_1,\tdone)^{\ppone\cdot\tdone(\strset_1)},(\tone_\epsilon,\tdone)^{\ppone\cdot\tdone(\strp\emptys)}\}$$
        and if we set
        $$\pdtwo'=\pdtwo\setminus\{(\ccase{\typeone}{\hole}{\tone_0}{1}{\tone_\epsilon},\tdtwo)^\pptwo\}
        +\{(\tone_0,\tdtwo)^{\pptwo\cdot\tdtwo(\strset_0)},(\tone_1,\tdtwo)^{\pptwo\cdot\tdtwo(\strset_1)},(\tone_\epsilon,\tdtwo)^{\pptwo\cdot\tdtwo(\strp\emptys)}\}$$
        we know
        $|\ppone\cdot\tdone(\strset)-\pptwo\cdot\tdtwo(\strset)|=\ppthree\cdot|\tdone(\strset)-\tdtwo(\strset)|\leq\ppthree\cdot\distone$,
        for all $\strset$, and then the thesis.
	\item The case $(\rrec{\typeone}{\contone}{\tone_0}{\tone_1}{\tone_\epsilon},\tdone)$ with $(\contone,\tdone)\in\values^\typestring$ is similar.  
    \item
        If the pair that reduces is
        $\pdone\ni(\tone,\tdone)^\ppone$ then we have that there
        exists $\pdtwo\ni(\tone,\tdtwo)^\pptwo$ with
        $|\ppone-\pptwo|\leq\ppthree\cdot\distone$. So if
        $\tone\onestep\{(\tone_i)^{\ppone_i}\}$ we have
        $\pdone'=\pdone\setminus\{(\tone,\tdone)^\ppone\}+\{(\tone_i,\tdone)^{\ppone\cdot\ppone_i}\}$;
        if we set
        $\pdtwo=\pdtwo\setminus\{(\tone,\tdtwo)^\pptwo\}+\{(\tone_i,\tdtwo)^{\pptwo\cdot\ppone_i}\}$
        and $\ppthree_i=\ppthree\cdot\ppone_i$ we have $\sum_i\ppthree_i=\ppthree$,
        $|\ppone_i\cdot\ppone-\ppone_i\cdot\pptwo|\leq\ppthree_i\cdot\distone$,
        so we get the thesis.
    \item
        If the pair that reduces is
        $\pdone\ni(\contone\vone,\tdone)^\ppone$ with
        $\contone\neq\lambda\varone.\contone'$, then we
        have by the one-step rules
        $(\contone\vone,\tdone)\onestep\{(\contone',\tdone')^1\}$.
        By definition of $\relative\distone$ there exists $\pdtwo\ni(\contone\vone,\tdtwo)^\pptwo$
        with $\ppone=\pptwo$ and
        $(\contone\vone,\tdtwo)\onestep\{(\contone',\tdtwo')^1\}$;
        so we have:
        \begin{align*}
            \pdone\pairconv\pdone'=\pdone\setminus\{((\contone\vone,\tdone)^\ppone)\}+\{(\contone',\tdone')^\ppone\}\\
            \pdtwo\pairconv\pdtwo'=\pdtwo\setminus\{((\contone\vone,\tdtwo)^\pptwo)\}+\{(\contone',\tdtwo')^\pptwo\}
        \end{align*}
        but by a previous lemma
        $\tramet(\tdone',\tdtwo')\leq\distone$ and so we have the
        thesis.
    \item
        Otherwise, if we have another pair that reduces
        $\supp(\pdone)\ni(\contone,\tdone)\onestep\{(\contone_i,\tdone_i)^{\ppone_i}\}$
        we have
        $\pdone'=\pdone\setminus\{(\contone,\tdone)^\ppone\}+\{(\contone_i,\tdone_i)^{\ppone\cdot\ppone_i}\}$
        if we set
        $\pdtwo'=\pdtwo\setminus\{(\contone,\tdtwo)^\pptwo\}+\{(\contone_i,\tdtwo_i)^{\pptwo\cdot\ppone_i}\}$
        we get the thesis. Indeed by a previous lemma we know
        $\tramet(\tdone_i,\tdtwo_i)\leq\distone$ and by definition
        $\ppone=\pptwo$ so
        $\ppone\cdot\ppone_i=\pptwo\cdot\ppone_i$ for all $i$.
    \end{enumerate}
\end{enumerate}

\end{proof}
\begin{lemma}[$d$-Relatedness, internally]\label{lemma:intdrel}
Given two pair distributions $\pdone,\pdtwo$ with $\pdone\relative\distone\pdtwo$ then there exist $\pdone',\pdtwo'$ normal pair distributions with $\pdone\smallstep\emptyt\pdone'$ or $\pdone\pairconv\pdone'$ and $\pdtwo\smallstep\emptyt\pdtwo'$ or $\pdtwo\pairconv\pdtwo'$ such that $\pdone'\relative\distone\pdtwo'$.
\end{lemma}

\begin{lemma}[$d$-Relatedness, externally]\label{lemma:extdrel}
Given two pair distributions $\pdone\relative\distone\pdtwo$ then for every trace $\trtwo$:
\begin{enumerate}
\item
	If $\trtwo$ doesn't end with the action $\view\cdot$ then there exist $\pdone',\pdtwo'$ normal pair distributions such that $\pdone\smallstep\trtwo\pdtwo'$ and $\pdtwo\smallstep\trtwo\pdtwo'$ with $\pdone'\relative\distone\pdtwo'$.
\item
    Otherwise if $\pdone\smallstepreal\trtwo\ppone_1$ and $\pdtwo\smallstepreal\trtwo\ppone_2$ we have $|\ \ppone_1-\ppone_2|\leq\distone$.
\end{enumerate}
\end{lemma}
\begin{proof}
We act by induction on the length of $\trtwo$ first by proving the first case and then the second one. 
\begin{enumerate}
\item
    If $\trtwo=\emptyt$ then we get the thesis by lemma~\ref{lemma:intdrel}.\\
    If $\trtwo=\trtwo'\cdot\pass\vone$ then by the small-step rules
    we have
    $\pdone\smallstep{\trtwo'}\pdone',\pdtwo\smallstep{\trtwo'}\pdtwo'$
    with
    $\pdone'=\{(\contone_i,\tdone_i)^{\ppone_i}\},\pdtwo'=\{(\contone_i,\tdtwo_i)^{\pptwo_i}\}$ normal pair distributions
    and by induction hypothesis we have
    $\pdone'\relative\distone\pdtwo'$.\\
    By the one-step rules we know that the action $\pass\vone$
    is allowed only if
    $\contone=\lambda\varone.\contone'$ or
    $\contone=\hole$ and
    $\tdone=\{(\lambda\varone.\tone_h)^{\pone_h}\}$. So if we set
    $J=\{j\in I\ |\ \contone_j=\lambda\varone.\contone_j'\},K=\{k\in I\ |\
    \contone_k=\hole\}$, by the small-step rules we have:
    \begin{align}
        \pdone'\smallstep{\pass\vone}\pdone''=\{(\contone'_j\subs{\vone}{\varone},\tdone_j)^{\ppone_j}\}+\{(\hole_k,\tdone_k')^{\ppone_k}\}\\
        \pdtwo'\smallstep{\pass\vone}\pdtwo''=\{(\contone'_j\subs{\vone}{\varone},\tdtwo_j)^{\pptwo_j}\}+\{(\hole_k,\tdtwo_k')^{\pptwo_k}\}
    \end{align}
    but by a previous lemma
    $\tramet(\tdone_k',\tdtwo_k')\leq\distone$ so we have
    $\pdone''\relative\distone\pdtwo''$ and by applying lemma~\ref{lemma:intdrel} we get the thesis.

\item
    If $\trtwo=\trtwo'\cdot\view\strset$ then
    we have by induction hypothesis $\pdone\smallstep{\trtwo'}\pdone'=\{(\contone_i,\tdone_i)^{\ppone_i}\},
    \pdtwo\smallstep{\trtwo'}\pdtwo'=\{(\contone_i,\tdtwo_i)^{\pptwo_i}\}$
    with $\pdone',\pdtwo'$ normal pair distributions and
    $\pdone'\relative\distone\pdtwo'$.\\
    By the small-step rules we know that
    $$\pdone'\smallstepreal{\view\strset}\ppone=\sum\ppone_i'\cdot\ppone_i$$
    with $(\contone_i,\tdone_i)\onestep^{\view\strset}\ppone_i'$ and
    $\pdone((\contone_i,\tdone_i))=\ppone_i$.\\
    Similarly
    $$\pdtwo'\smallstepreal{\view\strset}\pptwo=\sum\pptwo_i'\cdot\pptwo_i$$
    with
    $(\contone_i,\tdtwo_i)\onestep^{\view\strset}\pptwo_i'$ and
    $\pdtwo((\contone_i,\tdtwo_i))=\pptwo_i$.\\\\
    At this point we make a distinction: by the one-step rules we
    know that the action $\view\strset$ is allowed only if
    $\contone_i=\strp{\strone_i}$ or $\contone_i=\hole$ and the
    term distribution in the pair is a string distribution, so we
    set: $J=\{j\in I\ ;\ \contone_j=\strp{\strone_j}\}, K=\{k\in I\ ;\
    \contone_k=\hole\}$. Obviously we have $I=J+ K$.\\
    By the 1-step rules we have:
    \begin{align*}
        &(\strp{\strone_j},\tdone_j)\onestep^{\view\strset}
        \left\{%
        \begin{array}{ll}
            1, & \hbox{If $\strp{\strone_j}\in\strset$;} \\
            0, & \hbox{Otherwise.} \\
        \end{array}%
        \right.
        \qquad
        (\strp{\strone_j},\tdtwo_j)\onestep^{\view\strset}
        \left\{%
        \begin{array}{ll}
            1, & \hbox{If ${\strone_j}\in\strset$;} \\
            0, & \hbox{Otherwise.} \\
        \end{array}%
        \right.&\\
        \text{and}&\\
        &(\hole,\tdone_k)\onestep^{\view\strset}\tdone_k(\strset),\qquad
        (\hole,\tdtwo_k)\onestep^{\view\strset}\tdtwo_k(\strset)&\\
    \end{align*}

    So we have:
    \begin{align*}
        \ppone =& \sum_j\ppone_j\cdot\ppone_j' +
            \sum_k\ppone_k\cdot\tdone_k(\strset)\qquad
        \pptwo = \sum_j\pptwo_j\cdot\pptwo_j' +
            \sum_k\pptwo_k\cdot\tdtwo_k(\strset)
    \end{align*}
    and then:
    \begin{align*}
        |\ppone-\pptwo|=&
            |\sum_j\ppone_j\cdot\ppone_j'+\sum_k\ppone_k\cdot\ppone_k'-\sum_j\pptwo_j\cdot\pptwo_j'-\sum_k\pptwo_k\cdot\pptwo_k'|=\\
        =&
            |\sum_j\ppone_j\cdot\ppone_j'-\pptwo_j\cdot\pptwo_j'+\sum_k\ppone_k\cdot\ppone_k'-\pptwo_k\cdot\pptwo_k'|
    \end{align*}
    But, by the fact that $\pdone'\relative\distone\pdtwo'$ we can
    say:
    \begin{enumerate}
    \item
        $\ppone_j'=\pptwo_j'$, by the fact that $\contone_j={\strone_j}$ is the same for $\pdone'$ and $\pdtwo'$, so we call both $\ppthree_j'$.
    \item
        $|\ppone_j-\pptwo_j|\leq\ppthree_j\cdot\distone$. Furthermore $\ppthree'_j\cdot|\ppone_j-\pptwo_j|\leq\ppthree_j\cdot\distone$ because $\ppthree'_j\in\{0,1\}$.
    \item
        $\ppone_k=\pptwo_k=\ppthree_k$ for all $k$, because
        $\contone_k=\hole\neq\tone$.
    \item
    $|\tdone_k(\strset)-\tdtwo_k(\strset)|\leq\tramet(\tdone_k,\tdtwo_k)\leq\distone$ for all $k$.
    \end{enumerate}
    Therefore:
    \begin{align*}
        |\ppone-\pptwo|=&
            |\sum_j\ppthree_j'\cdot(\ppone_j-\pptwo_j)+\sum_k\ppthree_k\cdot(\tdone_k(\strset)-\tdtwo_k(\strset))|\leq\\
        \leq&
            |\sum_j\ppthree_j'\cdot(\ppone_j-\pptwo_j)|+|\sum_k\ppthree_k\cdot(\tdone_k(\strset)-\tdtwo_k(\strset))|\leq\\
        \leq&
            \sum_j\ppthree_j'\cdot|\ppone_j-\pptwo_j|+\sum_k\ppthree_k\cdot|\tdone_k(\strset)-\tdtwo_k(\strset)|\leq
            \sum_j\ppthree_j\cdot\distone+\sum_k\ppthree_k\cdot\distone=\\
        =&
            \sum_i\ppthree_i\cdot\distone=\distone\cdot\sum_i\ppthree_i=\distone
    \end{align*}
    and so the thesis.
\end{enumerate}
\end{proof}
\begin{myth}[Non-expansiveness]
Given two term distributions such that $\tramet(\tdone,\tdtwo)=\distone$ then for all contexts $\contone$ we have that $\tramet(\contone[\tdone],\contone[\tdtwo])\leq\distone$.
\end{myth}
\begin{proof}
In order to get the thesis we have to prove that for all traces $\trtwo$ we have that $|\ \Pr(\contone[\tdone],\trtwo)-\Pr(\contone[\tdtwo],\trtwo)\ |\leq\distone$.
\begin{itemize}
\item
	If $\trtwo$ doesn't end with the $\view\cdot$ action then we have: $|\ \Pr(\contone[\tdone],\trtwo)-\Pr(\contone[\tdtwo],\trtwo)\ | = |\ 1-1\ | = 0\leq \distone$
\item
	Otherwise we have that $(\contone,\tdone)\smallstepreal\trtwo\ppone_1=\Pr(\contone[\tdone],\trtwo)$ and similarly $(\contone, 	\tdtwo)\smallstepreal\trtwo\ppone_2=\Pr(\contone[\tdtwo],\trtwo)$. But it is clear that $\{(\contone,\tdone)^1\}\relative\distone\{(\contone,\tdtwo)^1\}$, and so by lemma~\ref{lemma:extdrel} we have: $|\ \ppone_1-\ppone_2\ |\leq\distone$ and then the thesis.
\end{itemize}
\end{proof}
\begin{myth}
For all $\tone,\ttwo$,$\conmet(\tone,\ttwo)\leq\tramet(\tone,\ttwo)$.
\end{myth}
\begin{proof}
By the previous theorem we know that, if $\tramet(\tone,\ttwo)=\distone$ then for all context $\contone$, we have $\tramet(\contone[\tone],\contone[\ttwo])\leq\distone$;
so:
\begin{align*}
	\conmet(\tone,\ttwo)=&
		\sup_\contone|\ \sem{\contone[\tone]}(\strp\emptys)-\sem{\contone[\ttwo]}(\strp\emptys)\ |=
		\sup_\contone|\ \Pr(\contone[\tone],\view{\strp\emptys})-\Pr(\contone[\ttwo],\view{\strp\emptys})\ |\leq\\
	\leq&
		\sup_\contone\tramet(\contone[\tone],\contone[\ttwo])\leq\distone\qquad\forall\contone
\end{align*}
\end{proof}
As a corollary of non-expansiveness, one gets that:
\begin{theorem}[Full Abstraction]\label{theo:metricsfa}
For all $\tone,\ttwo$, $\tramet(\tone,\ttwo)=\conmet(\tone,\ttwo)$.
\end{theorem}
\begin{proof}
$\tramet(\tone,\ttwo)\leq\conmet(\tone,\ttwo)$ because by the full
abstraction lemma for all traces $\trone$ there exists a context
$\contone_\trone$ such that
$\sem{\contone_\trone[\tone]}(\strp\emptys)=\Pr(\tone,\trone)$ and so the
quantification over contexts catches the quantification over traces.
The other inclusion, $\conmet(\tone,\ttwo)\leq\tramet(\tone,\ttwo)$,
is a consequence of non-expansiveness.
\end{proof}
One may wonder whether a coinductive notion of distance, sort of a
metric analogue to applicative bisimilarity, can be defined. The
answer is positive~\cite{Desharnais99}. It however suffers from the same problems
applicative bisimilarity has: in particular, it is not fully abstract.

\section{Computational Indistinguishability}\label{sect:ci}
\newcommand{\deone}{D}
\newcommand{\detwo}{E}
\newcommand{\algone}{\mathcal{A}}
In this section we show how our notions of equivalence and distance
relate to computational indistinguishability (CI in the following), a
key notion in modern cryptography.
\begin{definition}
Two distribution ensembles $\{\deone_n\}_{n\in\NN}$ and
$\{\detwo_n\}_{n\in\NN}$ (where both $\deone_n$ and $\detwo_n$ are
distributions on binary strings) are said to be \emph{computationally
  indistinguishable} iff for every PPT algorithm $\algone$ the
following quantity is a negligible\footnote{A negligible function is a
  function which tends to $0$ faster than any inverse polynomial
  (see~\cite{Goldreich2001} for more details).} function of $n\in\NN$:
$\left|\mathrm{Pr}_{x\leftarrow\deone_n}(\algone(x,1^n)=\epsilon)-\mathrm{Pr}_{x\leftarrow\detwo_n}(\algone(x,1^n)=\epsilon)\right|$.
\end{definition}
It is a well-known fact in cryptography that in the definition above,
$\algone$ can be assumed to sample from $x$ just \emph{once} without
altering the definition itself, provided the two involved ensembles
are efficiently computable~(\cite{Goldreich2001}, Theorem 3.2.6, page 108).
This is in contrast to the case of arbitrary ensembles~\cite{GoldreichSudan}.

The careful reader should have already spotted the similarity between
CI and the notion of context distance as given in
Section~\ref{sect:metrics}. There are some key differences, though:
\begin{varenumerate}
\item\label{point:parameter}
  While context distance is an \emph{absolute} notion of distance,
  CI depends on a parameter $n$, the so-called \emph{security parameter}.
\item
  In computational indistinguishability, one can compare distributions
  over \emph{strings}, while the context distance can evaluate how
  far terms of \emph{arbitrary} types are.
\end{varenumerate}
The discrepancy Point \ref{point:parameter} puts in evidence, however, 
can be easily overcome by turning the context distance into something
slightly more parametric.
\begin{definition}[Parametric Context Equivalence]\label{def:pce}
Given two terms $\tone,\ttwo$ such that
$\proves\tone,\ttwo:\aone\typestring\to\typeone$, we say that $\tone$ and
$\ttwo$ are \emph{parametrically context equivalent} iff for
every context $\contone$ such that
$\proves\contone[\proves\typeone]:\typestring$
we have that
$\myabs{\sem{\contone[\tone\secs]}(\strp\emptys)-\sem{\contone[\ttwo\secs]}(\strp\emptys)}$
is negligible in $\secp$.
\end{definition}
This way, we have obtained a characterization of CI:
\begin{theorem}\label{theo:civsce}
Let $\tone,\ttwo$ be two terms of type $\aone\typestring\to\typestring$. Then
$\tone,\ttwo$ are parametric context equivalent iff
the distribution ensembles $\{\sem{\tone\secs}\}_{\secp\in\NN}$ and
$\{\sem{\ttwo\secs}\}_{\secp\in\NN}$ are computationally indistinguishable.
\end{theorem}
Please observe that Theorem~\ref{theo:civsce} only deals with
terms of type $\aone\typestring\to\typestring$. The significance
of parametric context equivalence when instantiated to
terms of type $\aone\typestring\to\typeone$, where $\typeone$ is
a higher-order type, will be discussed in Section~\ref{sect:hoci} below.
\subsection{Computational Indistinguishability and Traces}
\newcommand{\nfone}{\mathit{negl}} 
How could traces capture the peculiar way parametric context
equivalence treats the security parameter? First of all, observe that,
in Definition~\ref{def:pce}, the security parameter is passed to the
term being tested \emph{without} any intervention from the
context. The most important difference, however, is that contexts are
objects which test \emph{families of terms} rather than terms. As a
consequence, the action $\view{\cdot}$ does not take strings or finite
sets of strings as arguments (as in equivalences or metrics), but
rather \emph{distinguishers}, namely closed \RSLR\ terms of type
$\aone\typestring\to\typestring$ that we denote with the metavariable $\ddone$. The probability that a term $\tone$
of type $\typestring$ satisfies one such action $\view{\ddone}$
is $\sum_{\strone}\sem{\tone}(\strp\strone)\cdot\sem{\ddone\strp\strone}(\strp\emptys)$.

A trace $\trone$ is said to be \emph{parametrically compatible} for a
type $\aone\typestring\to\typeone$ if it is compatible for
$\typeone$. This is the starting point for the following definition:
\begin{definition}
  Two terms $\tone,\ttwo:\typeone$ are \emph{parametrically trace
    equivalent}, we write $\tone\ptreq\ttwo$, iff for every trace
  $\trone$ which is parametrically compatible with $\typeone$, there
  is a negligible function $\nfone:\NN\rightarrow\RR_{[0,1]}$ such
  that
  $\myabs{\Pr(\tone,\pass{\secs}\cdot\trone)-\Pr(\ttwo,\pass{\secs}\cdot\trone)}\leq\nfone(\secp)$.
\end{definition}
The fact that parametric trace equivalence and parametric
context equivalence are strongly related is quite intuitive: they are
obtained by altering in a very similar way two notions which are
already known to coincide (by Theorem~\ref{theo:metricsfa}). Indeed:
\begin{theorem}\label{theo:plcevspte}
Parametric trace equivalence and parametric context equivalence
coincide.
\end{theorem}
The first inclusion is trivial, indeed every trace can be easily
emulated by a context.  The other one, as usual is more difficult, and
requires a careful analysis of the behavior of terms depending on
parameter, when put in a context. Overall, however, the structure of
the proof is similar to the one we presented in
Section~\ref{sect:traceequivalence}.
%
The first step towards the proof is the introduction of a particular class of
distinguishers $\ddone_{\strp\strone}$ such that:
\begin{align*}
\sem{\ddone_{\strp\strone}\strp\strone'}(\strp\emptys) = 
\left\{
\begin{array}{ll}
1 & \hbox{if $\strp\strone'=\strp\strone$}\\
0 & \hbox{otherwise}
\end{array}
\right.
\end{align*}
We formalize the use of a distinguisher as argument of the action $\view\cdot$ by giving the rules in Figure~\ref{figure:disting}.\\
\begin{figure}[!h]
\centering \fbox{
\begin{minipage}{.95\textwidth}
\footnotesize
{
	$$
	\AxiomC{$\sem{\ddone\strp\strone}(\strp\emptys)=\ppone$}
	\UnaryInfC{$(\strp\strone,\tdone)\onestep^{\view\ddone}\ppone$}
	\DisplayProof
	\qquad
	\AxiomC{$\begin{array}{cc}
		\tdone=\{(\strp{\strone_i})^{\ppone_i}\} &
		\sem{\ddone\strp{\strone_i}}(\strp\emptys)=\ppone'_i
		\end{array}$}
	\UnaryInfC{$(\hole,\tdone)\onestep^{\view\ddone}=\sum\ppone_i\cdot\ppone'_i$}
	\DisplayProof
	$$
	\vspace{6pt}
	$$
	\AxiomC{$\begin{array}{ccc}
	\sum_{\strp{\zero\strone}}(\contone,\tdone)\onestep^{\view{\ddone_{\strp{\zero\strone}}}}=\ppone_0 &
	\sum_{\strp{\one\strone}}(\contone,\tdone)\onestep^{\view{\ddone_{\strp{\one\strone}}}}=\ppone_1 &
	(\contone,\tdone)\onestep^{\view{\ddone_\epsilon}}\ppone_\epsilon
	\end{array}$}
	\UnaryInfC{$(\ccase{\typeone}{\contone}{\tone_0}{\tone_1}{\tone_\epsilon},\tdone)\onestep\{(\tone_0,\tdone)^{\ppone_0},(\tone_1,\tdone)^{\ppone_1},(\tone_\epsilon,\tdone)^{\ppone_\epsilon}\}$}
	\DisplayProof
	$$
	\vspace{6pt}
	$$
	\AxiomC{$(\contone,\tdone)\onestep^{\view{\ddone_{\strp\strone}}}\ppone_{\strp\strone}$}
	\UnaryInfC{$(\rrec{\typeone}{\contone}{\tone_0}{\tone_1}{\tone_\epsilon},\tdone)\onestep
	\substack{
	\{((\tone_0\strp\strone)(\rrec{\typeone}{\strp\strtwo}{\tone_0}{\tone_1}{\tone_\epsilon}),\tdone)^{\ppone_{\strp\strone}}\}_{\strp\strone=\strp{\zero\strtwo}} + \\
	\{((\tone_1\strp\strone)(\rrec{\typeone}{\strp\strtwo}{\tone_0}{\tone_1}{\tone_\epsilon}),\tdone)^{\ppone_{\strp\strone}}\}_{\strp\strone=\strp{\one\strtwo}} + \\
	\{(\tone_\epsilon,\tdone)^{\ppone_{\strp\emptys}}\}}$}
	\DisplayProof
	$$
	\vspace{6pt}
	$$
	\AxiomC{$\begin{array}{cc}
	(\contone,\tdone)\smallstep\trtwo\{(\contone_i,\tdone_i)^{\ppone_i}\}
	&
	(\contone_i,\tdone_i)\onestep^{\view\ddone}\ppone_i'
	\end{array}$}
	\UnaryInfC{$(\contone,\tdone)\smallstep{\trtwo\cdot\view\ddone}\sum\ppone_i\cdot\ppone_i'$}
	\DisplayProof
	$$
}
\end{minipage}}
\caption {Distinguisher 1-step and Small-step rules}\label{figure:disting}
\end{figure}
In order to prove that parametric trace equivalence and parametric context equivalence coincide, we have to do some improvements to our approach: differently from Section~\ref{sect:traceequivalence} we will show that if $\tone,\ttwo:\aone\typestring\to\typeone$ are parametrically trace equivalent, then for all context $\lambda\varone.\contone[\proves\typeone]:\atwo\typestring\to\typetwo$ then $\lambda\varone.\contone[\tone\varone]\ptreq\lambda\varone.\contone[\ttwo\varone]$. This change is made because it is essential that the context passes the right security parameter to the term which it is testing; furthermore we will adapt the prove starting from a couple $(\lambda\varone.\contone,\ptdone)$ where $\ptdone=\{\tdone^\secp\}_{\secp\in\Nat}$ is a parametric term distribution, i.e. a family of term distributions of the form $\tdone^\secp=\{(\tone_i\secs)^{\ppone_i}\}$.\\
The idea behind the prove is that starting from $\{(\lambda\varone.\contone,\ptdone)^1\},\{(\lambda\varone.\contone,\ptdtwo)^1\}$, after a sequence of internal/external reduction performed by the context and the environment,  the first reduction inside the hole is the pass of the security parameter $\secs$ which in our new setting coincide to the choice of $\tdone^\secp\in\ptdone,\tdtwo^\secp\in\ptdtwo$ according to $\secs$; at this point, if we prove that the two pair distribution are $d-$related, by the non-expansiveness we will get the thesis.

\begin{lemma}
Given $\ptdone=\{\tdone^\secp\},\ptdtwo=\{\tdtwo^\secp\}$, with $\tdone^\secp=\{(\tone\secs)^1\},\tdtwo^\secp=\{(\ttwo\secs)^1\}$, if $\tone\ptreq\ttwo$ then for all $\secp\in\Nat$ there exists $\neglone:\Nat\to\Real$ negligible, such that $\tramet(\tdone^\secp,\tdtwo^\secp)\leq\neglone(\secp)$ 
\end{lemma}
\begin{proof}
If $\tone\ptreq\ttwo$, then $\exists\neglone$ such that $|\Pr(\tone,\pass\secs\cdot\trone)-\Pr(\ttwo,\pass\secs\cdot\trone)|^\distone\leq\neglone(\secp)$. So we have that:
\begin{align*}
|\Pr(\tone\secs,\trone)-\Pr(\ttwo\secs,\trone)|=&
	|\Pr(\{(\tone\secs)^1\},\trone)-\Pr(\{(\ttwo\secs)^1\},\trone)|=\\
	& 
	|\Pr(\tdone^\secp,\trone)-\Pr(\tdtwo^\secp,\trone)|=\tramet(\tdone^\secp,\tdtwo^\secp)\leq\neglone(\secp)
\end{align*}
\end{proof}
\begin{theorem}[Parametric Congruence]\label{th:parcong}
Given two terms $\tone,\ttwo:\aone\typestring\to\typeone$ such that $\tone\ptreq\ttwo$, then for all context $\lambda\varone.\contone$ with $\proves\lambda\varone.\contone[\proves\typeone]:\typetwo$ we have: $\lambda\varone.\contone[\tone\varone]\ptreq\lambda\varone.\contone[\ttwo\varone]$.
\end{theorem}
\begin{proof}
Our goal is to prove that for all traces $\trone$ parametrically compatible with $\typetwo$ we have that there exists $\neglone:\NN\to\Real$ negligible such that: 
\begin{align*}
|\Pr(\lambda\varone.\contone[\tone\varone],\pass\secs\cdot\trone)-\Pr(\lambda\varone.\contone[\ttwo\varone],\pass\secs\cdot\trone)|\leq\neglone(\secp)
\end{align*}
We can see the terms inside the hole as parametric term distributions $\ptdone=\{(\tone\secs)^1\}_{\secp\in\NN},\ptdtwo=\{(\ttwo\secs)^1\}_{\secp\in\NN}$; so if we start from the pair distributions $\{(\lambda\varone.\contone,\ptdone)^1\}, \{(\lambda\varone.\contone,\ptdtwo)^1\}$ we have that the first reduction step is external, indeed the environment passes the value $\secs$, so we get:
\begin{align*}
\{(\lambda\varone.\contone,\ptdone)^1\}\onestep^{\pass\secs}\{(\contone\subs{\secs}{\varone},\ptdone)^1\}
\qquad
\{(\lambda\varone.\contone,\ptdtwo)^1\}\onestep^{\pass\secs}\{(\contone\subs{\secs}{\varone},\ptdtwo)^1\}
\end{align*}
At this point we can suppose that the context reduces internally and externally (depending on its type) so we split the trace $\trone$ in $\trone_1\cdot\trone_2$, where $\trone_1$ is the trace performed by the context; actually the fact is that it reduces in the same way for both pair distributions so, we can say that $\{(\contone\subs{\secs}{\varone},\ptdone)^1\}\smallstep{\trone_1}\{(\contone_i,\ptdone)^{\ppone_i}\},\{(\contone\subs{\secs}{\varone},\ptdtwo)\}\smallstep{\trone_1}\{(\contone_i,\ptdtwo)^{\ppone_i}\}$.\\
Now the only possible reduction is a term distribution reduction, i.e. a reduction inside the hole, but this means a choice of a term distribution inside the family depending on $\secp$; so we get:
\begin{align*}
\{(\contone_i,\ptdone)^{\ppone_i}\}\onestep\{(\contone_i,\{(\tone\secs)^1\})^{\ppone_i}\}
\qquad
\{(\contone_i,\ptdtwo)^{\ppone_i}\}\onestep\{(\contone_i,\{(\ttwo\secs)^1\})^{\ppone_i}\}
\end{align*}
But $\tone\ptreq\ttwo$, so by the previous lemma we have that there exists $\neglone:\NN\to\Real$ negligible such that $\tramet(\{(\tone\secs)^1\},\{(\ttwo\secs)^1\}\leq\neglone(\secp)$; furthermore it is obvious that $\{(\contone_i,\{(\tone\secs)^1\})^{\ppone_i}\}\relative\distone\{(\contone_i,\{(\ttwo\secs)^1\})^{\ppone_i}\}$ with $\distone\leq\neglone(\secp)$ and by applying the Lemma~\ref{lemma:extdrel} we have that for all traces $\trone_2$:
$$|\Pr(\{(\contone_i,\{(\tone\secs)^1\})^{\ppone_i}\},\trone_2)-\Pr(\{(\contone_i,\{(\ttwo\secs)^1\})^{\ppone_i}\},\trone_2)|\leq\distone\leq\neglone(\secp)$$
and this means:
\begin{align*}
|\Pr(\lambda\varone.\contone[\tone\varone],\pass\secs\cdot\trone)-\Pr(\lambda\varone.\contone[\ttwo\varone],\pass\secs\cdot\trone)|\leq\neglone(\secp)
\end{align*}
and then the thesis.
\end{proof}

\begin{corollary}
Given two terms $\tone,\ttwo:\aone\typestring\to\typeone$, if they are parametrically trace equivalent, then they are parametrically context equivalent.
\end{corollary}
\begin{proof}
For all context $\proves\contone[\typeone]:\typestring$ we have:
\begin{align*}
&|\sem{\contone[\tone\secs]}(\strp\emptys)-\sem{\contone[\ttwo\secs]}(\strp\emptys)|=
	|\sem{(\lambda\varone.\contone[\tone\varone])\secs}(\strp\emptys)-\sem{(\lambda\varone.\contone[\ttwo\varone])\secs}(\strp\emptys)|=\\
&
\qquad |\Pr(\lambda\varone.\contone[\tone\varone],\pass\secs\cdot\view{\ddone_\epsilon})-\Pr(\lambda\varone.\contone[\ttwo\varone],\pass\secs\cdot\view{\ddone_\epsilon})|\leq\neglone(\secp)
\end{align*}
where $\neglone:\NN\to\Real$ is a negligible function.
\end{proof}
\subsection{An Example}
We propose an example in which we analyze two different programs. Both
of them are functions of type
$\modal\typestring\to\modal\typestring\to\typestring$: the first
one returns the string received in input padded or cut depending on
its length and on the security parameter , the second one produces a
random string and compare it to the input (padded or cut). If the
comparison is negative it returns the input string (padded or cut),
otherwise it returns the opposite. We use some syntactic sugar in oder
to make the terms more understandable.
\begin{align*}
	\tone:=&
        \lambda\varsec.\lambda\varone.\level\ \varone\ \varsec\\ \ttwo:=&
        \lambda\varsec.\lambda\varone.\ite{(\level\ \varone\ \varsec)=(\rbg\ \varsec)}{\neg(\level\ \varone\ \varsec)}{(\level\ \varone\ \varsec)}
\end{align*}
The function $\level$ receives in input two strings and pads or cuts
the first one in order to return a string of the same length as the
second one received in input, the function $\rbg$ returns a random
string of the length of the one received in input and the function
$\neg$ switches all the bits of the string in input. So, for all
$\secp\in\Nat,\strp\strone\in\values^\typestring$, if we set
$\strp\strone'=\level\ \strp\strone\ \secs$ we have:
\begin{align*}
	\tone\smallstep{\pass{\secs}\cdot\pass{\strp\strone}}&
		\{(\strp\strone')^1\}\\
	\ttwo\smallstep{\pass{\secs}\cdot\pass{\strp\strone}}&
		\{(\strp\strone')^{1-\ppone},(\neg\strp\strone')^\ppone\}
\end{align*}
Where $\ppone=\Pr[\strp\strone' = \rbg\ \secs]=\frac{1}{2^\secp}$.
So for all distinguisher $\ddone$, if we set $\sem{\ddone\strp\strone'}=\ppone_1,\sem{\ddone(\neg\strp\strone')}=\ppone_2$ we have that:
\begin{align*}
\tone & \smallstepreal{\pass\secs\cdot\pass{\strp\strone}\cdot\view\ddone}\ppone_1\\
\ttwo & \smallstepreal{\pass\secs\cdot\pass{\strp\strone}\cdot\view\ddone}(1-\ppone)\cdot\ppone_1 + \ppone\cdot\ppone_2
\end{align*}
And so we have $\tone,\ttwo$ are parametric trace equivalent, indeed for all $\secp,\strp\strone,\ddone$ we have:
\begin{align*}
	|\Pr(\tone,\pass{\secs}\cdot\pass{\strp\strone}\cdot\view\ddone)-\Pr(\ttwo,\pass{\secs}\cdot\pass{\strp\strone}\cdot\view\ddone)|=\\
	|\ppone_1-((1-\ppone)\cdot\ppone_1+\ppone\cdot\ppone_2)| = |\ppone_1-\ppone_1+\ppone\cdot\ppone_1-\ppone\cdot\ppone_2| = \ppone\cdot|\ppone_1-\ppone_2|\leq\ppone=\frac{1}{2^\secp}
\end{align*}
which is negligible. This, in particular, implies that the two terms
are parametrically trace equivalent, thus parametrically context equivalent.
\subsection{Higher-Order Computational Indistinguishability?}\label{sect:hoci}
Theorem~\ref{theo:civsce} and Theorem~\ref{theo:plcevspte} together
tell us that two terms $\tone,\ttwo$ of type
$\aone\typestring\to\typestring$ are parametrically trace equivalent
iff the distributions they denote are computationally
indistinguishable.  But what happens if the type of the two terms
$\tone,\ttwo$ is in the form $\aone\typestring\to\typeone$ where
$\typeone$ is an \emph{higher-order} type? What do we obtain?
Actually, the literature on cryptography does not offer a precise
definition of ``higher-order'' computational indistinguishability, so
a formal comparison with parametric context equivalence is not
possible, yet.

Apparently, linear contexts do not capture equivalences as
traditionally employed in cryptography, already when $\typeone$ is the
first-order type $\aone\typestring\to\typestring$. A central concept
in cryptography, indeed, is \emph{pseudorandomness}, which can be
spelled out for strings, giving rise to the concept of a pseudorandom
generator, but also for functions, giving rise to pseudorandom
functions~\cite{KatzLindell2007}.  Formally, a function
$F:\{0,1\}^*\to\{0,1\}^*\to\{0,1\}^*$ is said to be a pseudorandom
function iff $F(s)$ is a function which is indistinguishable from a
random function from $\{0,1\}^n$ to $\{0,1\}^n$ whenever $s$ is drawn
at random from $n$-bit strings. Indistinguishability, again, is
defined in terms of PPT algorithms having \emph{oracle} access to
$F(s)$.  Now, having access to an oracle for a function is of course
different than having \emph{linear} access to it. Indeed, building a
\emph{linear} pseudorandom function is very easy: $G(s)$ is defined
to be the function which returns $s$ independently on
the value of its input. $G$ is of course not pseudorandom in the
classical sense, since testing the function multiple times a
distinguisher immediately sees the difference with a truly random
function. On the other hand, the \RSLR\ term $\tone_G$ implementing
the function $G$ above is such that $\lambda\varone.\tone_G\ttwo$ is
trace \emph{equivalent} to a term $\tthree$ where:
\begin{varitemize}
\item
  $\ttwo$ is a term which produces in output $|\varone|$ bits drawn
  at random;
\item
  $\tthree$ is the term $\lambda\varone.\tfour$
  of type $\aone\typestring\to\atwo\typestring\to\typestring$
  such that $\tfour$ returns a random function from $|\varone|$-bitstrings
  to $|\varone|$-bitstrings. Strictly speaking, $\tthree$ cannot
  be an \RSLR\ term, but it can anyway be used as an idealized construction.
\end{varitemize}
But this is not the end of the story. Sometime, enforcing linear
access to primitives is necessary. Consider, as an example, the two
terms
$$
\tone=\lambda n.(\lambda k.\lambda x.\lambda y.Enc(x,k))Gen(n)\quad
\ttwo=\lambda n.(\lambda k.\lambda x.\lambda y.Enc(y,k))Gen(n)
$$
where $Enc$ is meant to be an encryption function and $Gen$ is a
function generating a random key. $\tone$ and $\ttwo$ hould be considered
equivalent whenever $Enc$ is a secure cryptoscheme. But if $Enc$ is
secure against passive attacks (but not against active attacks), the
two terms can possibly be distinguished with high probability if
copying is available. The two terms can indeed be proved to be
parametrically context equivalent if $Enc$ is the cryptoscheme induced
by a pseudorandom generator.

Summing up, parametrized context equivalence coincides with CI when
instantiated on base types, has some interest also on higher-order
types, but is different from the kind of equivalences cryptographers
use when dealing with higher-order objects (e.g. when defining
pseudorandom functions). This discrepancy is mainly due to the linearity
of the contexts we consider here. It seems however very hard to
overcome it by just considering arbitrary nonlinear contexts instead
of linear ones.  Indeed, it would be hard to encode any
\emph{arbitrary} PPT distinguisher accessing an oracle by an
\RSLR\ context: those adversaries are only required to be PPT for
oracles implementing certain kinds of functions (e.g. $n$-bits to
$n$-bits, as in the case of pseudorandomness), while filling a
\RSLR\ context with any PPT algorithm is guaranteed to result in a PPT
algorithm. This is anyway a very interesting problem, which is outside
the scope of this paper, and that we are currently investigating in
the context of a different, more expressive, probabilistic
$\lambda$-calculus.
\section{Conclusions}
In this paper, we have studied notions of equivalence and metrics in a
language for higher-order probabilistic polytime computation. More
specifically, we have shown that the discriminating power of linear
contexts can be captured by traces, both when equivalences and metrics
are considered. Finally, we gave evidence on how applicative bisimilarity
is a sound, but not fully abstract,  methodology for context equivalence.

We believe, however, that the main contribution of this work is the
new light it sheds on the relations between computational
indistinguishability, linear contexts and traces. In particular, this
approach, which is implicitly used in the literature on the
subject~\cite{Zhang10,NowakZ10}, is shown to have some limitations,
but also to suggest a notion of higher-order indistinguishability
which could possibly be an object of study in itself. This is indeed
the main direction for future work we foresee.

\bibliographystyle{abbrv}
\bibliography{biblio}
\end{document}